\documentclass[a4paper,UKenglish,cleveref, autoref, thm-restate]{lipics-v2021}

\usepackage{hyperref}
\usepackage{comment}

\usepackage{stfloats}

\usepackage{graphicx} 

\usepackage{amsmath,amsfonts,amssymb,xspace}
\usepackage{bbm}
\usepackage{graphicx}
\usepackage{paralist}
\usepackage{xcolor}
\usepackage{multicol}

\usepackage{algorithm}
\usepackage[noend]{algpseudocode}
\usepackage{xifthen}
\usepackage{float}
\usepackage{adjustbox}

\usepackage{algorithmicx}%

\nolinenumbers
\pdfoutput=1 
\hideLIPIcs  

\newcommand{\macrocolor}{.}
\newcommand{\mathmacro}[1]{{\ensuremath{\textcolor{\macrocolor}{#1}}}}

\def\eps{\mathmacro{\varepsilon}}
\def\bR{\mathmacro{\mathbb R}}
\def\bN{\mathmacro{\mathbb N}}

\def\bX{\mathmacro{\mathbb X}}
\def\cD{\mathmacro{\mathcal{D}}}

\def\localG{\mathmacro{\mathcal{L}}}
\def\globalG{\mathmacro{\mathcal{G}}}
\def\redGlobG{\mathmacro{\widetilde{\globalG}}}
\def\unifyG{\mathmacro{\mathcal{U}}}

\renewcommand{\emph}[1]{\textit{\textbf{#1}}}
\newcommand{\rev}[1]{\mathmacro{\mathrm{rev}({#1})}}
\newcommand{\candW}[1]{\mathmacro{\widehat{w}_{#1}}}
\newcommand{\molW}[1]{\mathmacro{{w}_{#1}}}
\newcommand{\sweepseq}{\mathmacro{\mathfrak{s}}}
\newcommand{\edgeseqs}[1]{\mathmacro{\mathfrak{S}_{#1}}}

\newcommand{\dfree}[3][\Delta]{%
    \ifthenelse{\equal{#2}{} }
    {\mathmacro{\cD_{#1}}}
    {\mathmacro{\cD^{(#2,#3)}_{#1}}}
}

\newcommand{\kopt}{\mathmacro{{k_\Delta}}}

\newcommand{\ConcreteCandidates}[1]{\mathmacro{{{C_{#1}}}}}

\newcommand{\df}{\mathmacro{{\mathrm{d}_\mathcal{F}}}}
\newcommand{\Cov}{\mathmacro{{\mathrm{Cov}}}}
\newcommand{\proxyCov}{\mathmacro{{\widehat{\Cov}}}}
\newcommand{\curvespace}[1]{\mathmacro{{\mathbb{X}^d_{#1}}}}
\newcommand{\atomic}[1]{\mathmacro{{\mathcal{A}_{#1}}}}
\newcommand{\molecular}[1]{\mathmacro{{\mathcal{M}_{#1}}}}
\newcommand{\extremal}[1]{\mathmacro{{\mathcal{E}_{#1}}}}

\newcommand{\combinatorial}{\mathmacro{{\mathcal{R}}}}
\newcommand{\bad}{\mathmacro{{\mathcal{B}}}}

\newcommand{\nonprob}{\mathmacro{{\widehat{\mathcal{S}}}}}

\renewcommand{\O}{\mathmacro{{\mathcal{O}}}}
\newcommand{\tildeO}{\mathmacro{{\Tilde{\mathcal{O}}}}}
\newcommand{\coarsity}{{\mathmacro{{\alpha}}}}

\makeatletter
\def\moverlay{\mathpalette\mov@rlay}
\def\mov@rlay#1#2{\leavevmode\vtop{%
   \baselineskip\z@skip \lineskiplimit-\maxdimen
   \ialign{\hfil$\m@th#1##$\hfil\cr#2\crcr}}}
\newcommand{\charfusion}[3][\mathord]{
    #1{\ifx#1\mathop\vphantom{#2}\fi
        \mathpalette\mov@rlay{#2\cr#3}
      }
    \ifx#1\mathop\expandafter\displaylimits\fi}
\makeatother

\newcommand{\anne}[1]{\textcolor{blue}{Anne: #1}}

\title{Subtrajectory Clustering and Coverage Maximization in Cubic Time, or Better}

\author{Jacobus Conradi}{University of Bonn, Bonn, Germany}{jacobus.conradi@gmx.de}{https://orcid.org/0000-0002-8259-1187}{Partially funded by the Deutsche Forschungsgemeinschaft (DFG, German Research Foundation) - 313421352 (FOR 2535 Anticipating Human Behavior) and the iBehave Network: Sponsored by the Ministry of Culture and Science of the State of North Rhine-Westphalia. Affiliated with Lamarr Institute for Machine Learning and Artificial Intelligence.}
\author{Anne Driemel}{University of Bonn, Germany}{driemel@cs.uni-bonn.de}{https://orcid.org/0000-0002-1943-2589}{Affiliated with Lamarr Institute for Machine Learning and Artificial Intelligence.}

\authorrunning{Conradi, and Driemel}

\Copyright{Jacobus Conradi, Anne Driemel} 



\begin{document}

\addtocounter{linenumber}{1000}
\maketitle

\begin{abstract}


Many application areas collect unstructured trajectory data. In subtrajectory clustering, one is interested to find patterns in this data using a hybrid combination of segmentation and clustering. We analyze two variants of this problem based on the well-known \textsc{SetCover} and \textsc{CoverageMaximization} problems. In both variants the set system is induced by metric balls under the Fréchet distance centered at polygonal curves. Our algorithms focus on improving the running time of the update step of the generic greedy algorithm by means of a careful combination of sweeps through a candidate space.  In the first variant, we are given a polygonal curve $P$ of complexity $n$, distance threshold $\Delta$ and complexity bound $\ell$ and the goal is to identify a minimum-size set of center curves $\mathcal{C}$, where each center curve is of complexity at most $\ell$ and every point $p$ on $P$ is covered. A point $p$ on $P$ is covered if it is part of a subtrajectory $\pi_p$ of $P$ such that there is a center $c\in\mathcal{C}$ whose Fréchet distance to $\pi_p$ is at most $\Delta$. 
We present an approximation algorithm for this problem with a running time of $\O((n^2\ell + \sqrt{k_\Delta}n^{5/2})\log^2n)$, where $k_\Delta$ is the size of an optimal solution. The algorithm gives a bicriterial approximation guarantee that relaxes the Fréchet distance threshold by a constant factor and the size of the solution by a factor of $\O(\log n)$. 
The second problem variant asks for the maximum fraction of the  input curve $P$ that can be covered using $k$ center curves, where $k\leq n$ is a parameter to the algorithm. 
For the second problem variant, one can show that our techniques lead to an algorithm with a running time of $\O((k+\ell)n^2\log^2 n)$ and similar approximation guarantees. Note that in both algorithms $k,k_\Delta\in O(n)$ and hence the running time is cubic, or better if $k\ll n$.
\end{abstract}

\newpage
\setcounter{tocdepth}{2}
\tableofcontents

\newpage

\section{Introduction} 
Trajectory analysis is a broad field ranging from gait analysis of full-body motion trajectories~\cite{lee2002gait,Qiao2017RealtimeHG,ionescu2013human3}, via traffic analysis~\cite{acmsurvey20} and epidemiological hotspot detection~\cite{HotspotSurvey} to Lagrangian analysis of particle simulations~\cite{VANSEBILLE201849}. Using the current availability of cheap tracking technology, vast amounts of unstructured spatio-temporal data are collected. Clustering is a fundamental problem in this area and can be formulated under various similarity metrics, among them the Fréchet distance. However, the unstructured nature of trajectory data poses the additional challenge that the data first needs to be segmented before the trajectory segments, which we call subtrajectories, can be clustered in a metric space. Subtrajectory clustering aims to find clusters of subtrajectories that represent complex patterns that reoccur along the trajectories~\cite{agarwal2018,buchinGroup20,LiangYWLCXL24}, such as commuting patterns in traffic data.

We study two variants of a problem posed by Akitaya, Brüning, Chambers, and Driemel~\cite{Akitaya2021Covering}, which are based on the well-known set cover and coverage maximization problems. Both rely on the definition of a geometric set system built using metric balls under the Fréchet distance. 
Given a polygonal curve, we are asked to produce a set of `center' curves that either

\begin{compactenum}
    \item has minimum cardinality and covers the entirety of the input curve, or
    \item covers as much as possible of the input curve under cardinality constraints.
\end{compactenum}

In either setting a point $p$ of the input-curve is considered as `covered', if there is a subcurve $\pi_p$ of the input curve that containts $p$ and has small Fréchet distance to some center curve. Note that a point $p$ may be covered by multiple center curves. For the precise formulation refer to \Cref{sec:probdef}. This formulation extends immediately to a set of input curves---instead of one---where one is now tasked to cover the entirety (respectively as much as possible) of all points on all input curves combined.

In this paper we study the continuous variant of the problem where we need to cover the continuous set of points given by the input polygonal curve.
It is instructive, however, to first consider a discrete variant of the set cover instance as follows. Given a discrete point sequence $A=a_1,\dots, a_n$, consider the set of contiguous subsequences of the form $A_{ij}=a_i,\dots,a_j$. We say $a_k$ is contained in the set defined by $A_{ij}$ if there exist $i'\leq k \leq j'$ such that $A_{i'j'}$ is similar to $A_{ij}$. Note that a full specification of this set system has size $\Theta(n^3)$ in the worst case. 
This phenomenon is shared with the continuous setting where all known discretizations lead to a cubic-size set system. In this paper, we examine the question if we can obtain subcubic-time algorithms by internally representing the set system and any partial solution implicitly.


\subsection{Basic Notation and Definitions}

An \emph{edge} is a map $e:[0,1]\rightarrow\bR^d$ that is defined by two points $p,q\in\bR^d$ as the linear interpolation $t\mapsto p + t(q-p)$. A \emph{polygonal curve} is a map $P:[0,1]\rightarrow\bR^d$ defined by an ordered set $(p_1,\ldots,p_n)\subset\bR^d$ as the result of concatenating the edges $(e_1,\ldots,e_{n-1})$, where $e_i$ is the edge from $p_i$ to $p_{i+1}$. Any point $P(t)$ is said to lie on the edge $e_i$ if $t$ lies in the preimage of $e_i$. Any point $P(t)$ lies on exactly one edge, except if it defines the endpoint of edge $e_i$ which coincides with the start point of $e_{i+1}$. The number of points $n$ defining the polygonal curve $P$ is its \emph{complexity}, and we denote it by $|P|$. The set of all polygonal curves in $\bR^d$ of complexity at most $\ell$ we denote by $\bX^d_\ell$. For any polygonal curve $P$ and values $0\leq s\leq t\leq 1$ we define $P[s,t]$ to be the subcurve of $P$ from $s$ to $t$, which itself is a polygonal curve of complexity at most $|P|$.

For two polygonal curves $P$ and $Q$ we define the \emph{continuous Fréchet distance} to be 
\[\df(P,Q)=\inf_{f,g}\max_{t\in[0,1]}\|P(f(t))-Q(g(t))\|\]
where $f$ and $g$ range over all non-decreasing surjective functions from $[0,1]$ to $[0,1]$.

We follow the definition introduced in \cite{Akitaya2021Covering} defining the notion of $\Delta$-Coverage. Let $P$ be a polygonal curve in $\bR^d$, and let $\Delta\in\bR_{>0}$ and $\ell\in\bN_{\geq 2}$ be given. For any $\mathcal{Q}\subset\bX^d_\ell$, define the $\Delta$-coverage of $\mathcal{Q}$ on $P$ as
\[\Cov_P(\mathcal{Q},\Delta) = \bigcup_{Q\in\mathcal{Q}}\left(\bigcup_{0\leq s\leq t\leq 1,\,\df(P[s,t],Q)\leq\Delta}[s,t]\right)\subset[0,1].\]
For a curve $Q\in\bX^d_\ell$ we may denote the $\Delta$-coverage $\Cov_P(\{Q\},\Delta)$ by $\Cov_P(Q,\Delta)$.

\subsection{Problem definition}\label{sec:probdef}

\begin{figure*}
    \centering
    \includegraphics{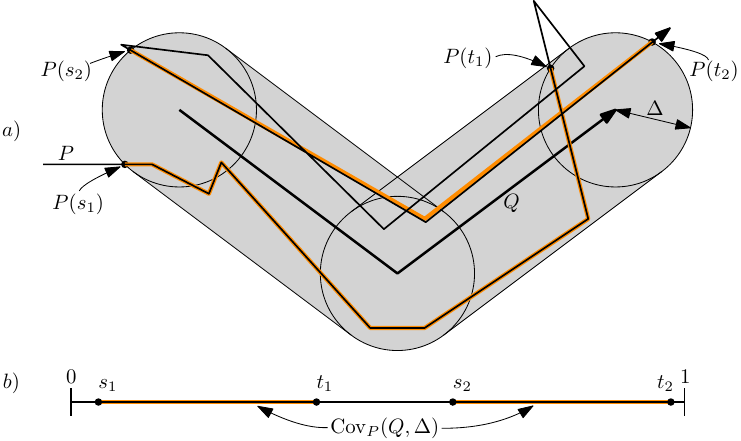}
    \caption{$a)$: Example of all points on $P$ that lie on subcurves of $P$ that have Fréchet distance at most $\Delta$ to a curve $Q$ of complexity $3$. $b)$: The set $\Cov_P(Q,\Delta)\subset[0,1]$.}
    \label{fig:coverage}
\end{figure*}

\subparagraph{Problem 1: Subtrajectory Covering (SC)}
Let $P$ be a polygonal curve in $\bR^d$, and let $\Delta\in\bR_{>0}$ and $\ell\in\bN_{\geq 2}$ be given. The task is to identify a set $\mathcal{Q}$ of curves in $\bX^d_\ell$ minimizing $|\mathcal{Q}|$ such that \(\Cov_P(\mathcal{Q},\Delta)=[0,1]\).

\begin{definition}[Bicriterial approximation for SC]
Let $P$ be a polygonal curve in $\bR^d$, and let $\Delta\in\bR_{>0}$ and $\ell\in\bN_{\geq 2}$ be given. Let $C^*\subset\curvespace{\ell}$ be a set of minimum cardinality, such that $\Cov_P(C^*,\Delta)=[0,1]$. An algorithm that outputs a set $C\subset\curvespace{\ell}$ of size $\alpha|C^*|$ such that $\Cov_P(C,\beta\Delta)=[0,1]$ is called an $(\alpha,\beta)$-approximation for SC.
\end{definition}

\subparagraph{Problem 2: Subtrajectory Coverage Maximization (SCM)}
Let $P$ be a polygonal curve in $\bR^d$, and let $\Delta\in\bR_{>0}$, $\ell\in\bN_{\geq 2}$ and $k\in\bN_{\geq 1}$ be given. The task is to identify a set $\mathcal{Q}$ of $k$ curves in $\bX^d_\ell$ maximizing the Lebesgue-measure \(\left\|\Cov_P(\mathcal{Q},\Delta)\right\|\).

\begin{definition}[Bicriterial approximation for SCM]
Let $P$ be a polygonal curve in $\bR^d$, and let $\Delta\in\bR_{>0}$ and $\ell\in\bN_{\geq 2}$ be given. Let $C^*\subset\curvespace{\ell}$ be a set of size $k$ such that $\left\|\Cov_P(C^*,\Delta)\right\|$ is maximal. An algorithm that outputs a set $C\subset\curvespace{\ell}$ of size $k$ such that 
\(\left\|\Cov_P(C,\beta\Delta)\right\|\geq \alpha\left\|\Cov_P(C^*,\Delta)\right\|\)
is called an $(\alpha,\beta)$-approximation for SCM.
\end{definition}

\subsection{Prior work on Subtrajectory Covering and Coverage Maximization}\label{sec:priorwork}

We first discuss prior work on the SC problem.
The SC problem was introduced by Akitaya, Brüning, Chambers and Driemel \cite{Akitaya2021Covering}.  Their results were later 
improved by Brüning, Conradi, and Driemel \cite{Brüning2022Faster}. 
Both works identified a curve $S$ that approximates the input $P$, such that any solution to the problem induces an approximate solution of similar size consisting of only subcurves of $S$. 
This set of subcurves defines a set system which turns out to have low VC-dimension. This enables randomized set cover algorithms with improved approximation factors and expected running time in $\O(n(\kopt)^3(\log^4(\Lambda/\Delta)+\log^3(n/\kopt))+n\log^2n)$, where $\Lambda$ corresponds to the arc length of the curve $P$. 
While the complexity of the main algorithm in \cite{Akitaya2021Covering} depends on the relative arclength of the input curve, the latter \cite{Brüning2022Faster} also identified a set of points on $S$ that define $\O(n^3)$ subcurves of complexity $2$ which induce the set-system of the \textsc{SetCover} instance resulting in an expected running time of $\tildeO(\kopt n^3)$, where $\tildeO(\cdot)$ hides polylogarithmic factors. 
The approximation quality of both approaches depends on $\ell$, $\kopt$ as well as on the VC-dimension of the set system. 
A drawback of these algorithms is that they generate cluster centers that are line segments only and that the algorithm is randomized with large constants in the approximation factors. Conradi and Driemel \cite{conradi2023findingcomplexpatternstrajectory} took a different approach by focusing on the greedy set cover algorithm, which successively adds the candidate curve maximizing the added coverage to a partial solution until the entirety of $[0,1]$ is covered. 
Their algorithm computes $\O(n^3\ell)$ candidate curves that have complexity up to $\ell$ instead of $2$. 
Applying the greedy algorithm, this results in a deterministic $(\O(\log n),11)$-approximation algorithm with running time $\tildeO(\kopt n^4\ell+n^4\ell^2)$. 
Recently, van der Hoog, van der Horst and Ophelders showed~\cite{vanderhoog2024fasterdeterministicsubtrajectoryclustering}  that the cardinality of the candidate curve set can be further reduced to a size of $\O(n^2\log n)$. 
They  improve the quality of the approximating curve $S$ resulting in a $(\O(\log n),4)$-approximation with running time in $\tildeO(\kopt n^3)$.

As for the SCM problem, by standard sub-modular function maximization arguments \cite{Nemhauser1978,krause11}, the greedy  algorithm used in all these approaches \cite{Akitaya2021Covering,Brüning2022Faster,conradi2023findingcomplexpatternstrajectory,vanderhoog2024fasterdeterministicsubtrajectoryclustering} yields an $(\O(1),\O(1))$-approximation. 
Using this reduction, the candidate curve set identified in \cite{vanderhoog2024fasterdeterministicsubtrajectoryclustering} yields a $(\O(1),4)$-approximation to the SCM problem with a running  time of $\tildeO(k n^3)$.

\subsection{Structure and results}

We start in \Cref{sec:techniques} by describing general algorithmic techniques that are applicable to variants of subtrajectory clustering when using the greedy algorithm that subsequently finds the largest remaining cluster of elements not covered so far. In the remaining sections we use the techniques developed in \Cref{sec:techniques} and apply them to the two problems Subtrajectory Covering and Subtrajectory Coverage Maximization. In the following, we give an overview of the individual sections and how they lead to the main results.

\subsubsection{Sweep-sequences}
In \Cref{sec:approxfreespace,sec:setsystem} we describe how to obtain a set of candidate center curves $\mathcal{C}$ based on known transformations. 
Unlike the recent line of work discussed above \cite{Akitaya2021Covering,Brüning2022Faster,conradi2023findingcomplexpatternstrajectory,vanderhoog2024fasterdeterministicsubtrajectoryclustering}, we do not improve any further on the size of the set $\mathcal{C}$. Instead, our focus lies on the update step of the greedy set cover algorithm. In this step, we have to find the element $c\in \mathcal{C}$ that maximizes the coverage that is added to a partial solution. 
As the coverage of any element in $\mathcal{C}$ may consist of $\Theta(n)$ disjoint intervals and the set of candidates has quadratic size \cite{vanderhoog2024fasterdeterministicsubtrajectoryclustering}, one may be tempted to believe that the running time per round should be at least cubic in $n$. We manage to show that this step can be performed faster if the coverage is maintained implicitly. 
To this end, we identify a new structure, which we call \emph{sweep-sequences}, in \Cref{sec:sweepseq}. 
The goal is to reduce the inherently two-dimensional search-space of subcurves of $S$ to few one-dimensional search-spaces that may be processed via a sequence of sweep algorithms.
In \Cref{sec:combRep,sec:symbolicRepresentation}, we describe how sweep-sequences allow efficient maintenance of the $\Delta$-coverage and with \Cref{thm:pointQuery} describe the main interface to this identified structure enabling our algorithm: for a weighted set $A\subset[0,1]$ of points we are able to compute the total weight of points that lie inside the $\Delta$-coverage of every element in the sweep-sequence.

\subsubsection{Subtrajectory Covering}
In \Cref{sec:cubicClusterin} and \Cref{sec:subcubic} we present two $(\O(\log n),4)$-approximation algorithms for Subtrajecotry Covering. The algorithms discretize the ground-set using known techniques via the arrangement (of size $\tildeO(n^3)$) of the set of $\O(n^2\log n)$  coverages induced by the elements in $\mathcal{C}$.
We obtain a first approximation algorithm for SC with running time in $\tildeO(|\mathcal{C}|n + \kopt|\mathcal{C}|)$. The algorithm first computes the arrangement in $\tildeO(|\mathcal{C}|n)$. Afterwards, each round of the greedy algorithm runs in (roughly) logarithmic time per element in $\mathcal{C}$ by repeatedly querying the subroutine from \Cref{thm:pointQuery}. In \Cref{sec:subcubic} we describe an improvement to this algorithm. Instead of computing the arrangement explicitly, we first identify a representative subset of size roughly $\O(\sqrt{\kopt}n^{3/2})$ of the arrangement which we cover in an initial pass. The solution covering this representative subset already covers almost every element in the arrangement. The remaining roughly $\O(\sqrt{\kopt}n^{5/2})$ elements of the arrangement are then explicitly identified and covered in a second pass with the algorithm from \Cref{sec:cubicClusterin} resulting in the following theorem.

\begin{restatable}{theorem}{coverafast}
\label{thm:coverafast}
    There is a $(96\ln(n)+128,4)$-approximation for SC. Given a polygonal curve $P$ of complexity $n$, $\Delta>0$ and $\ell\leq n$, its running time is in $\O\left(\left(n^2\ell+\sqrt{\kopt}n^{\frac{5}{2}}\right)\log^2n\right)$, where $\kopt$ is the size of the smallest subset $C^*\subset\bX^d_\ell$ such that $\Cov_P(C^*,\Delta)=[0,1]$.
\end{restatable}

As trivially $\kopt\leq\lceil\frac{n}{\ell}\rceil$, \Cref{thm:coverafast} yields an algorithm with (near-)cubic running time. Further, in the case that $\kopt\in\O(n^{1-\eps})$, it yields an algorithm with subcubic running time. This compares favorably to the best known algorithm with running time $\tildeO(\kopt n^3)$~\cite{vanderhoog2024fasterdeterministicsubtrajectoryclustering}.

\subsubsection{Subtrajectory Coverage Maximization}

In \Cref{sec:maximization} we show how the Lebesgue-measure of the coverage of the elements in $\mathcal{C}$ can be approximated efficiently by few piece-wise linear functions. This approximation can also be maintained efficiently by using the identified sweep-sequences from \Cref{sec:combRep,sec:symbolicRepresentation} allowing the evaluation of the aforementioned Lebesgue-measure of the coverage of every element in $\mathcal{C}$ in total time $\tildeO(|C|)$. This results in the following theorem, which compares favorably to the best known algorithm with running time $\tildeO(k n^3)$~\cite{vanderhoog2024fasterdeterministicsubtrajectoryclustering}.

\begin{restatable}{theorem}{thmMainCoverage}
\label{thm:maincoverage}
    Let $\eps>0$. There is a $(\frac{e-1}{16e},4+\eps)$-approximation algorithm for SCM. Given a polygonal curve $P$ of complexity $n$, $\Delta>0$, $\ell\leq n$ and $k>0$, its running time is in $\O((k+\ell)n^2\eps^{-2}\log ^{2} n\log \eps^{-1})$.
\end{restatable}

\subsection{Related work}

Prior to the line of work on the SC and SCM problem discussed in \Cref{sec:priorwork}, Buchin et al.~\cite{BuchinBGLL11} presented one of the earlier works on subtrajectory clustering for both the discrete and continuous Fréchet distance. Their work focuses on finding the largest subtrajectory cluster, where different variants of `largest' are considered. 
They present hardness-results for $(2-\eps)$-approximations for any $\eps$ and a matching polynomial time $2$-approximation algorithm. Gudmundsson and Wong~\cite{abs-2110-15554} present a cubic lower-bound conditioned on SETH for the same problem and show that this lower bound is tight. Note that the condition that the subtrajectories are disjoint is essential to their lower bound construction, while the SC and SCM problems do not have this constraint. The formulation by Buchin et al.~\cite{BuchinBGLL11} was also studied in subsequent work~\cite{GudmundssonV15,buchinGroup2017,BuchinBGHSSSSSW20, buchinGroup20}. 

Agarwal et al.~\cite{agarwal2018} formulate the subtrajectory clustering problem differently, namely based on \textsc{FacilityLocation}. In this problem formulation  there is an opening cost associated with every center curve, and there is a cost for every point on the input that is assigned to a cluster, and a different cost for points that are not assigned. They show conditional NP-hardness results but also give an $O(\log^2 n)$-approximation for well-behaved classes of curves under the discrete Fréchet distance.

\section{Algorithmic toolbox}\label{sec:techniques}
In this section we describe the underlying techniques and subroutines that allow the design of efficient algorithms for the SC and SCM problems.

\subsection{Approximate Free Space}\label{sec:approxfreespace}

We now turn to the central geometric tool that enables many of the techniques we will use lateron, namely the free space diagram and $\Delta$-free space. We generalize the notion of $\Delta$-free spaces slightly to make our techniques applicable for both the SC and SCM problems.

\begin{figure}
    \centering
    \includegraphics[width=\linewidth]{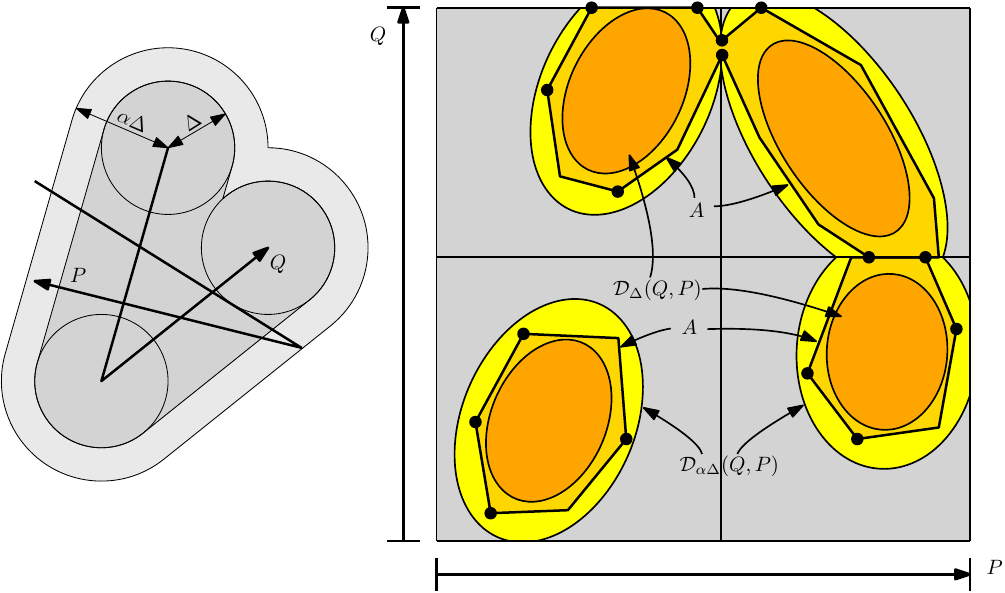}
    \caption{$\Delta$-free space and $\alpha\Delta$-free space of two curves $P$ and $Q$, as well as an $\alpha$-approximate $\Delta$-free space $A$ of $P$ and $Q$. Additionally marked are the extremal points of $A$.}
    \label{fig:free-space}
\end{figure}

\begin{definition}[Free space diagram]
	Let $S$ and $P$ be two polygonal curves parametrized over $[0,1]$.
	The free space diagram of $S$ and $P$ is defined as the joint parametric space $[0,1]^2$ together with a not necessarily uniform grid, where each vertical grid line corresponds to a vertex of $P$ and each horizontal grid line to a vertex of $S$.
	The $\Delta$-\emph{free space} (refer to \Cref{fig:free-space}) of $S$ and $P$ is defined as \[ \dfree{}{}(S,P) = \left\{(x,y)\in[0,1]^2 \mid \|P(x) -S(y)\|\leq\Delta \right\}. \] 
	This is the set of points in the parametric space, whose corresponding points on $S$ and $P$ are at a distance at most $\Delta$.
	The edges of $S$ and $P$ segment the free space diagram into cells.
	We call the intersection of $\dfree{}{}(S,P)$ with the boundary of cells the \emph{free space intervals}.
\end{definition}%

As was noted by Alt and Godau~\cite{AltG95}, the Fréchet distance of $S$ and $P$ is at most $\Delta$, if and only if there is a monotone (in both $x$- and $y$-direction) path from $(0,0)$ to $(1,1)$ in $\Delta$-free space of $S$ and $P$. They further observed that the free space inside any cell is described by an ellipse intersected with the cell and thus is convex and has constant complexity. Observe that two subcurves $P[a,c]$ and $S[b,d]$ have Fréchet distance at most $\Delta$ if and only if there is a monotone (in both $x$- and $y$-direction) path from $(a,b)$ to $(c,d)$ in $\Delta$-free space of $S$ and $P$.

\begin{definition}[Approximate free space]
	Let $S$ and $P$ be two polygonal curves parametrized over $[0,1]$. Let $\Delta\geq 0$ and $\alpha\geq 1$ be given. We say $A$ is an $\alpha$-approximate $\Delta$-free space (or $(\alpha,\Delta)$-free space) of $S$ and $P$, if 
    \begin{compactenum}
        \item $\dfree{}{}(S,P)\subset A\subset\dfree[\alpha\Delta]{}{}(S,P)$,
        \item $A$ is convex inside the interior of every cell of the free space diagram and
        \item any point $p$ on the boundary of a cell of the free space diagram is in $A$ iff for every cell $p$ lies in it is in the closure of the restriction of $A$ to the interior of that cell.
    \end{compactenum}
\end{definition}

The $(1,\Delta)$-free space of $S$ and $P$ is unique and coincides with $\dfree{}{}(S,P)$. By definition any approximate free space of $S$ and $P$ also defines an approximate free space of $\rev{S}$ and $P$ by mirroring it along the $x$-axis. Further, by definition $A$ is not only convex in the interior of every cell but also in the entirety of every cell. Further, similar to free space intervals we also define free space intervals of an approximate free space diagram by the intersection of the approximate free space diagram with the boundary of each cell.

Any monotone path from $(a,b)$ to $(c,d)$ in an $(\alpha,\Delta)$-free space of $S$ and $P$ implies that $\df(P[a,c],S[b,d])\leq\alpha\Delta$. Thus given an $(\alpha,\Delta)$-free space $A$ of $S$ and $P$ we can approximate the $\Delta$-coverage, namely let $\Cov_A(S[b,d])$ be the union of all intervals $[a,c]$, such that there is a monotone path from $(a,b)$ to $(c,d)$ inside $A$, then 
\[\Cov_P(S[b,d],\Delta)\subset\Cov_A(S[b,d])\subset\Cov_P(S[b,d],\alpha\Delta).\]

We extend the notion of extremal points---i.e. local minima and maxima of each cell of the $\Delta$-free space---from \cite{Brüning2022Faster,conradi2023findingcomplexpatternstrajectory,vanderhoog2024fasterdeterministicsubtrajectoryclustering} to approximate free spaces.

\begin{definition}[Extremal points \cite{Brüning2022Faster}]
    Let $A$ be an $(\alpha,\Delta)$-free space. As $A$ is convex in any cell $C$, the set of points of $A$ minimizing the $x$-coordinate in $C$ are described by a vertical line-segment of length at least $0$. We call the start- and end-point of this line-segment the left-most points of $A$ in $C$. Similarly $A$ has at most two bottom-most, rightmost and top-most points in $C$. The union over all left-, bottom-, right- and top-most points of $A$ in every cell defines the set of extremal points of $A$.
\end{definition}

Using symbolic perturbation arguments we may assume that no two extremal points share the same $y$-coordinate and the $y$-coordinates of top-most points is bigger than the $y$-coordinates of left- and right-most points of the same cell, and similarly the $y$-coordinate of bottom-most points is less than the $y$-coordinates of left- and right-most points of the same cell.

Aside from the $(1,\Delta)$-free space we will make use of an approximate $\Delta$-free space in this work, which we describe in the following.

\subsubsection{\boldmath Piecewise linear approximate $\Delta$-free space}

We provide an approximation to the $\Delta$-free space that is \textit{independent} of the dimension of the ambient space the polygonal curves live in. More precisely, the goal of this section is the following theorem.

\begin{theorem}\label{thm:approxFreeSpace}
    Let $P$ and $Q$ be two polygonal curves in $\bR^d$. In $O(|P||Q|\eps^{-2}\log(\eps^{-1}))$ one can compute a $(1+\eps,\Delta)$-free space of $P$ and $Q$ that consists of a convex polygon of complexity $O(\eps^{-2})$ in the interior of every cell.
\end{theorem}

We prove this via the following two underlying folklore lemmas. For completeness sake we prove them here.

\begin{lemma}\label{lem:folklore1}
    Let $\frac{1}{5}\geq \eps>0$. One can compute a polytope $D$ of complexity $\O(\eps^{-2})$ in $\bR^3$ such that the ball $B_1(0)$ of radius $1$ centered at $0$ in $\bR^d$ it holds that $B_1(0)\subset D\subset B_{1+4\eps}(0)$ in time $\O(\eps^{-2}\log \eps^{-1})$.
\end{lemma}
\begin{proof}
    Consider a grid with side length $\eps^{-1}$ in $\bR^3$. Observe that in $\O(\eps^{-2}\log\eps^{-1})$ one can compute all $\O(\eps^{-2})$ grid points that lie inside $B_1(0)$ such that there is a neighboring grid point that lies outside $B_1(0)$. Let now $C$ be the convex hull of these points. First observe that $C\subset B_1(0)$. Further every gridpoint in $B_1(0)$ is in $C$.
    
    For any point $x\in B_1(0)$ the value $\min_{y\in C}\|x-y\|< 2\eps^{-1}$. Otherwise there is a disk of radius $\eps^{-1}$ contained in $B_1(0)\setminus C$. But as the side length of the grid is $\eps^{-1}$ there must be a grid point $p$ in this disk of radius $\eps^{-1}$ and thus also in $B_1(0)\setminus C$. But then $p$ is also in $C$. Hence $B_{1-2\eps}(0)\subset C\subset B_1(0)$.

    Scaling $C$ by a factor of $(1+4\eps)$ implies the claim, as $1\leq1+2\eps -8\eps^2 =  (1+4\eps)(1-2\eps)$ and $ (1+4\eps)(1-2\eps)\leq (1+4\eps)$.
\end{proof}

\begin{lemma}\label{lem:folklore2}
    Let $e_1$ and $e_2$ be two edges in $\bR^d$ endowed with a polyhedral norm of complexity $C$. Then the $\Delta$-free space of $e_1$ and $e_2$ is a polygon of complexity $\O(C)$ and can be computed in $\O(C\log C)$ time.
\end{lemma}
\begin{proof}
    Let $H_1,\ldots,H_C$ be the half-planes described via their normal vector $n_i$ and value $d_i$ (that is $H_i=\{x\mid \langle n_i,x\rangle \leq  d_i\}$) in $\bR^d$ describing the ball of radius $\Delta$ under the polyhedral norm. Then the $\Delta$-free space is described as the intersection of the $C$ halfplanes (intersected with the cell of the free space) in $\bR^2$ such that $\langle n_i,e_1(x)-e_2(y)\rangle\leq d_i$.
\end{proof}

\begin{proof}[Proof of \Cref{thm:approxFreeSpace}]
    This is an immediate consequence of \Cref{lem:folklore1} and \Cref{lem:folklore2}. First in $O(|P||Q|\eps^{-2}\log\eps^{-1})$ compute for every cell of the free space diagram of $P$ and $Q$ a polygon that $(1+\eps)$-approximate the $\Delta$-free space. These define the $(1+\eps)$-approximate $\Delta$-free space in the interior of every cell. Hence by definition, on the shared boundary of a set of cells we define the $(1+\eps)$-approximate $\Delta$-free space as the intersection of all polygons associated to cells in that set.
\end{proof}

\subsection{Simplifications and Candidate Center Curves}
\label{sec:setsystem}

We now describe the following known result: Given a polygonal curve $P$ of complexity $n$, there is a polygonal curve $S$ of complexity $n$ and a set $\ConcreteCandidates{S}(P)$ of at most $\O(n^2\log n)$ subcurves of $S$ such that for any set of curves $C\subset\curvespace{\ell}$ there is a subset $\hat{C}\subset\ConcreteCandidates{S}(P)$ of size $\O(|C|)$ such that \(\Cov_P(C,\Delta)\subset\Cov_P(\hat{C},4\Delta)\).
This result follows the line of research in \cite{Brüning2022Faster,conradi2023findingcomplexpatternstrajectory,vanderhoog2024fasterdeterministicsubtrajectoryclustering} where with each successive work the multiplicative loss in the radius has been decreased from $11$ all the way down to $4$. We will describe the results with an additional level of abstraction, namely with \textit{approximate} free spaces instead of the usual notion of free spaces, that allows generalization to both the SC and SCM problem.

All of these known approaches \cite{Brüning2022Faster,conradi2023findingcomplexpatternstrajectory,vanderhoog2024fasterdeterministicsubtrajectoryclustering} compute a `maximally simplified' version of $P$, see also~\cite{de2013fast}. Our definition of a simplification is inspired by the notion of \textit{pathlet-preserving simplifications} from \cite{vanderhoog2024fasterdeterministicsubtrajectoryclustering}. We choose this variant as it yields the best constants.

\begin{definition}\label{def:goodsimp}
    Let $P$ be a polygonal curve in $\bR^d$. We call a curve $S$ a simplification of $P$ if $\df(S,P)\leq 2\Delta$ for any curve $Q$ of complexity $\ell$ and any subcurve $P[s,t]$ with $\df(Q,P[s,t])\leq\Delta$ there is a subcurve $S[s',t']$ such that $\df(P[s,t],S[s',t'])\leq 2\Delta$ and $|S[s',t']|\leq\ell+2$.
\end{definition}

\begin{theorem}[\cite{vanderhoog2024fasterdeterministicsubtrajectoryclustering}]\label{thm:thijs-simplification}
    For any curve $P$ of complexity $n$ a simplification $S$ of $P$ of complexity $\O(n)$ can be computed in $\O(n\log n)$ time.
\end{theorem}

\begin{observation}
    Let $P$ be a polygonal curve and let $S$ be a simplification of $P$. By definition of $S$ and the triangle inequality for any curve $\pi\in\curvespace{\ell}$ there is a subcurve $\hat{\pi}$ of $S$ of complexity at most $\ell+2$---that is $\hat{\pi}\in\curvespace{\ell+2}$---such that $\Cov_P(\pi,\Delta)\subset\Cov_P(\hat{\pi},4\Delta)$.
\end{observation}

\begin{definition}[Subcurve types]
Let $S$ be a polygonal curve. Let $\pi$ be a subcurve of $S$. We say $\pi$ is a \emph{vertex-vertex-subcurve} of $S$ if $\pi$ starts and ends at vertices of $S$. We say $\pi$ is a \emph{subedge} of $S$, if $\pi$ has complexity $2$, i.e. it is a subcurve of a single edge of $S$. We say $\pi$ is an \emph{edge-affix} of $S$ if it is a subedge of $S$ which starts or ends at a vertex of $S$.
\end{definition}

Based on (approximate) free spaces we now turn to defining the set of candidate center curves that will be the basis for all our algorithms. This set of candidate curves will serve as a small proxy set on which it suffices to approximately solve SC/SCM, instead of the entirety of $\curvespace{\ell}$.

\begin{figure}
    \centering
    \includegraphics[width=\linewidth]{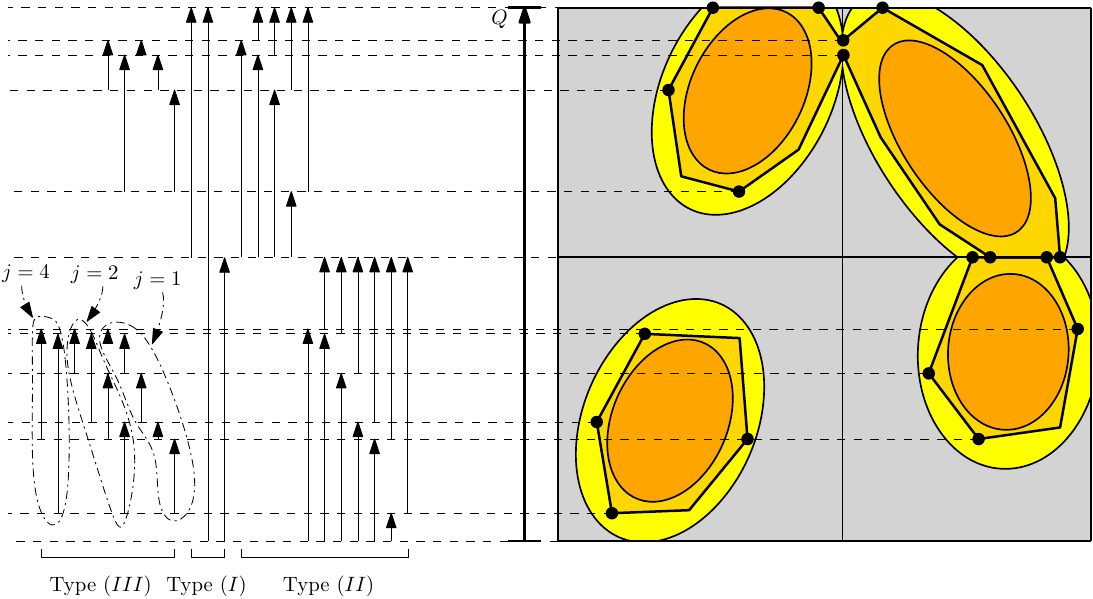}
    \caption{Illustration of all Type $(I)$-, $(II)$- and $(III)$-subcurves of $Q$ (as vertical lines) that are not reversals, induced by the approximate free space of $Q$ and $P$ from \Cref{fig:free-space}. Further marked are the values $j$, which induced the set of Type $(III)$-subcurves on the first edge of $Q$.}
    \label{fig:type2and3}
\end{figure}

\begin{definition}[Type $(I)$-, $(II)$- and $(III)$-subcurves]
    Let $P$ be a polygonal curve and let $S$ be a simplification of $P$. Let $\alpha$, $\Delta$ and $\ell$ be given. Let $A_S$ be an $(\alpha,\Delta)$-free space of $S$ and $P$. For every edge $e$ of $S$ let $A_e$ be the restriction of $A_S$ onto the edge $e$---that is $A_e$ is an $(\alpha,\Delta)$-free space of $e$ and $P$---and let $\extremal{}(A_e)\subset[0,1]=\{\eps_1,\ldots\}$ be a finite sorted superset of the $y$-coordinates of extremal points of $A_e$ (refer to \Cref{fig:free-space}).
Define three types of subcurves of $S$ via the sets $\extremal{}(A_e)$.
\begin{compactenum}[$(I)$:]
    \item A $(I)$-subcurve of $S$ is a vertex-vertex subcurve of $S$ that starts at some $i$th vertex of $S$ and ends at the $(i+j)$th vertex for $j\in\{2^m\mid0\leq m\leq \lfloor\log_2(\ell)\rfloor\}$.
    \item A $(II)$-subcurve of $S$ is either a affix-subcurve $e[0,\eps_i]$ or $e[\eps_i,1]$ of some edge $e$ of $S$ that is defined by a vertex of $S$ and some value $\eps_i\in\extremal{}(A_e)$ or its reversal $\rev{e[0,\eps_i]}$ or $\rev{e[\eps_i,1]}$.
    \item A $(III)$-subcurve of $S$ is either a subedge-subcurve $e[\eps_i,\eps_{i+j}]$ of an edge $e$ of $S$ that is defined by two values $\eps_i,\eps_{i+j}\in\extremal{}(A_e)$, such that $j\in\{2^m\mid0\leq m\leq \lfloor\log_2(|\extremal{}(A_e)|)\rfloor\}$, or its reversal $\rev{e[\eps_i,\eps_{i+j}]}$.
\end{compactenum}
\end{definition}

\begin{lemma}[\cite{vanderhoog2024fasterdeterministicsubtrajectoryclustering}]\label{lem:vanderhoogsplit}
    Let $P$ be a polygonal curve and let $S$ be a simplification of $P$. Let $\hat{\pi}$ be a subcurve of $S$ of complexity at most $\ell+2$. There is a set $S'_\pi$ consisting of either one subedge of $S$ or one vertex-vertex-subcurve of complexity at most $\ell$ and two edge-affixes of $S$ such that 
    \[\Cov_P(\hat{\pi},4\Delta)\subset\bigcup_{s\in S'_\pi}\Cov_P(s,4\Delta).\]
\end{lemma}

The proofs of the following lemma and theorem are adaptations of existing proofs extending them to the $(\alpha,\Delta)$-free spaces $A_S$. For the proofs of \Cref{lem:foursplit,lem:supersetfreespace} we refer to \Cref{apx:setsystemproofs}.

\begin{restatable}[adapted from \cite{Brüning2022Faster,conradi2023findingcomplexpatternstrajectory}]{lemma}{lemmaapxOne}\label{lem:foursplit}
    Let $P$ be a polygonal curve and let $S$ be a simplification of $P$. Let $A_S$ be an $(\alpha,\Delta)$-free space of $S$ and $P$. For every edge $e$ of $S$ let $A_e$ be the restriction of $A_S$ onto the edge $e$ and let $\extremal{}(A_e)\subset[0,1]=\{\eps_1,\ldots\}$ be a finite sorted superset of the $y$-coordinates of extremal points of $A_e$. 
    Let $\hat{\pi}$ be a subedge of the edge $e$. Then there is a set $S''_\pi$ consisting of at most four subedges of $S$ starting and ending at values in $\extremal{}(A_e)$ such that 
    \[\Cov_P(\hat{\pi},4\Delta)\subset\bigcup_{s\in S''_\pi}\Cov_{A_S}(s).\]
    If $\hat{\pi}$ is an edge-affix, then $S''_\pi$ consists of only two edge-affixes starting and ending at values in $\extremal{}(A_e)$.
\end{restatable}

\begin{restatable}[adapted from \cite{vanderhoog2024fasterdeterministicsubtrajectoryclustering}]{theorem}{lemmaapxTwo}\label{lem:supersetfreespace}
    Let $P$ be a polygonal curve and let $S$ be a simplification of $P$. Let $A_S$ be an $(\alpha,\Delta)$-free space of $S$ and $P$. For every edge $e$ of $S$ let $A_e$ be the restriction of $A_S$ onto the edge $e$ and let $\mathcal{E}(A_e)$ be a superset of the $y$-coordinates of extremal points of $A_e$ which induce a set $C\subset\curvespace{\ell}$ of Type $(I)$-, $(II)$- and $(III)$-subcurves of $S$. Then for any curve $\pi\in\curvespace{\ell}$ there is a set $S_\pi\subset C$ of size at most $8$, such that 
    \[\Cov_P(\pi,\Delta)\subset\bigcup_{s\in S_\pi}\Cov_{A_S}(s)\subset\bigcup_{s\in S_\pi}\Cov_P(s,4\alpha\Delta).\]
\end{restatable}

Overall, \Cref{thm:thijs-simplification,lem:supersetfreespace} (roughly) yield the following: Given a polygonal curve $P$, one can compute in $\tildeO(n^2)$ time a simplification $S$ and a set of $\tildeO(n^2)$ subcurves $C$ of $S$ such that for any solution to SC (resp. SCM) with $\Delta$ and $\ell$ of cardinality $k$ there is a subset of $C$ of cardinality $8k$ that solves SC (resp. SCM) with $4\Delta$ and $\ell$. Hence, approximately solving SC (resp. SCM) restricted to this small discrete set $C$ yields a bicriterial approximation.

\subsection{Sweep-sequences}\label{sec:sweepseq}

In this section we identify an ordering of the Type $(II)$ and Type $(III)$ subcurves from \Cref{lem:supersetfreespace} which allows maintaining a symbolic representation of (an approximation of) the coverage, as well as construction of a data structure that allows efficient point-queries returning the set of all Type $(II)$ and Type $(III)$ whose coverage includes the query point.

So far we have handled approximate free spaces in an abstract sense. To actually work with the approximate free space, we endow it with some additional information and make some assumptions.

\begin{definition}\label{def:cellfunc}
    Let $e$ be an edge and let $P$ be a polygonal curve. Let $A_e$ be an $(\alpha,\Delta)$-free space of $e$ and $P$. Then define $l_i(\cdot)$ as the function mapping any $y$ to the $x$-coordinate of the leftmost point of the $i$th cell in $A_e$ at height $y$. If this point does not exist, $l_i(y)=\infty$. Similarly define $r_j(y)$ to be the $x$-coordinate of the rightmost point at height $y$, and $\infty$ otherwise.
\end{definition}

Throughout this section we are given a polygonal curve $P$, values $\alpha,\Delta$ and $\ell$, as well as a simplification $S$ of $P$. Let $A_S$ be an $(\alpha,\Delta)$-free space of $S$ and $P$. For every edge $e$ of $S$ let $A_e$ be the restriction of $A_S$ onto the edge $e$. Further, let a superset $\extremal{}(A_e)$ of the $y$-coordinates of the extremal points of $A_e$ be given, which induce a set $\ConcreteCandidates{S}(P)$ of $(I)$-$(III)$-subcurves of $S$. We assume that for every edge $e$ the functions $l_i(\cdot)$ and $r_i(\cdot)$ of the free space $A_e$ can be computed in $\O(1)$ time. We further assume that $|\extremal{}(A_e)|=\O(n)$. Hence there are $\O(n\log\ell)$ Type $(I)$, $\O(n^2)$ Type $(II)$-subcurves and $\O(n^2\log n)$ Type $(III)$-subcurves in $\ConcreteCandidates{S}(P)$. 

\begin{figure}
    \centering
    \includegraphics[width=\linewidth]{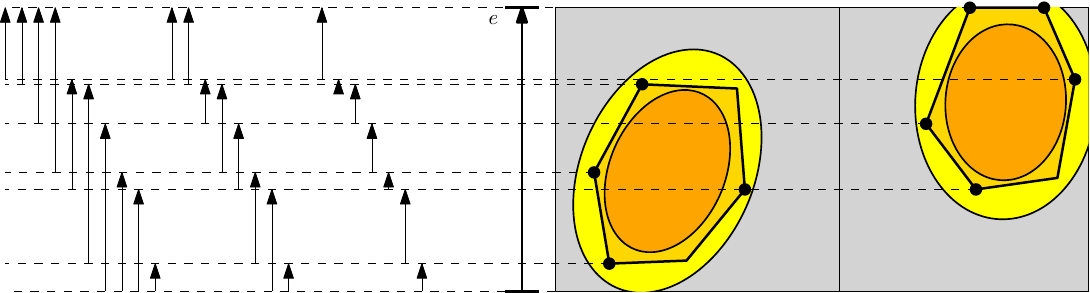}
    \caption{Illustration of three of the eight sweep-sequences in $\edgeseqs{e}$ of the edge $e$ that are constructed for Type $(III)$-subcurves. One for $j=1$, one for $j=2$ and one for $j=4$.}
    \label{fig:sweepsequence}
\end{figure}

\begin{definition}[Sweep-sequence]
Let $E=(e_1,\ldots)$ be a sorted list of values in $\bR$. We say $\sweepseq$ is a sweep-sequence of $E$ if $\sweepseq$ is a list of tuples of $E$ such that either for all consecutive tuples $(e_a,e_b)$ and $(e_c,e_d)$ it holds that $a\leq b$, $c\leq d$, $0\leq a-c\leq1$ and $0\leq b-d\leq 1$ or for all consecutive tuples it holds that $a\geq b$, $c\geq d$, $0\leq c-a\leq1$ and $0\leq d-b\leq 1$.
\end{definition}

\begin{lemma}\label{lem:allsweep}
    Let $e$ be an edge of $S$. Assume $0,1\in\extremal{}(A_e)$. There is a set $\edgeseqs{e}$ of $\O(\log n)$ sweep-sequences of $\extremal{}(A_e)$, each of length $\O(n)$, such that for any Type $(II)$ or Type $(III)$-subcurve $\pi$ on the edge $e$ there is a tuple $(\eps_i,\eps_j)$ in one of the sweep-sequences in $\edgeseqs{e}$ such that $\pi=e[\eps_i,\eps_j]$ if $\eps_i\leq \eps_j$, and $\pi=\rev{e[\eps_j,\eps_i]}$ if $\eps_i>\eps_j$. Further, for the first and last pair $(e_a,e_b)$ in every sweep-sequence it holds that $a=b$.
\end{lemma}
\begin{proof}
    Let $\extremal{}(A_e)=\{0=\eps_1,\ldots,\eps_m=1\}$.
    We construct one sweep-sequence for the Type $(II)$-subcurves, and $\O(\log n)$ sweep-sequences for the Type $(III)$-subcurves. 

    For the Type $(II)$-subcurves the sought after sweep-sequence is precisely the sequence $\{(\eps_1,\eps_1),(\eps_1,\eps_2),\ldots,(\eps_1,\eps_{m-1}),(\eps_1,\eps_m),(\eps_{2},\eps_m),\ldots,(\eps_m,\eps_m)\}$.

    For the Type $(III)$-subcurves we construct two sweep-sequences for any $j\in\{2^m\mid 0\leq m\leq\lfloor\log_2(|\extremal{}(A_e)|)\rfloor\}$, namely the sequences $\{(\eps_1,\eps_{1+2^j}),(\eps_2,\eps_{2+2^j}),\ldots,(\eps_{m-2^j},\eps_{m})\}$ and $\{(\eps_{1+2^j},\eps_1),(\eps_{2+2^j},\eps_2),\ldots,(\eps_{m},\eps_{m-2^j})\}$. To satisfy the requirement that the first and last pair consists of two copies of the same value, we append and prepend the sequence $\{(\eps_1,\eps_1),(\eps_1,\eps_2),\ldots,(\eps_1,\eps_{2^j})\}$ (resp.\ $\{(\eps_1,\eps_1),(\eps_2,\eps_1),\ldots,(\eps_{2^j},\eps_1)\}$) and append the sequence $\{(\eps_{m-2^j+1},\eps_{m}),\ldots,(\eps_m,\eps_m)\}$ (resp.\ $\{(\eps_{m},\eps_{m-2^j+1}),\ldots,(\eps_m,\eps_m)\}$) (refer to \Cref{fig:sweepsequence}).

    The claim then follows by definition of Type $(II)$- and $(III)$-subcurves as well as the fact that $m=|\extremal{}(A_e)|\in\O(n)$.
\end{proof}

Observe that for any edge $e$ of $S$ these sweep-sequences $\edgeseqs{e}$ can be constructed in $\O(n\log n)$ time.

In the remainder of this section we focus our analysis on sweep-sequences where for every tuple $(\eps_a,\eps_b)$ it holds that $a\leq b$. Any analysis of such sweep-sequences carries over to sweep-sequences where for every tuple $(\eps_a,\eps_b)$ it holds that $a>b$, by setting $e\gets\rev{e}$, and thus to all sweep-sequences constructed in \Cref{lem:allsweep}.

\subsection{Combinatorial representation and the proxy coverage}\label{sec:combRep}

In this section we define an approximation to the notion of $\Delta$-coverage as well as an underlying structure, the \textit{reduced combinatorial representation}, which we later on show how to maintain efficiently. An important consequence is the definition of the \textit{proxy coverage} (\Cref{def:proxycov}) which allows efficient maintenance along a sweep-sequence.

\begin{definition}[Combinatorial representation]
    Let $e$ be an edge of $S$ and let $0\leq s\leq t\leq 1$ be given.
    The combinatorial representation $\combinatorial(e[s,t])$ of the coverage $\Cov_{A_S}(e[s,t])$ of $e[s,t]$ is the set of all inclusionwise-maximal pairs of indices $(i,j)$, such that there are points $s'$ and $t'$ on the $i$th and $j$th edge of $P$ respectively, such that there is a monotone path from $(s',s)$ to $(t',t)$ in $A_S$. A pair of indices $(i,j)$ includes another pair $(i',j')$ if $i\leq i'$ and $j'\leq j$.
\end{definition}

The combinatorial representation separates into two sets, the global group $\globalG(e[s,t])$ consisting of all index-pairs $(i,j)\in\combinatorial(e[s,t])$ such that $i<j$ and the local group $\localG(e[s,t])$ consisting of all index-pairs $(i,i)\in\combinatorial(e[s,t])$.

\begin{observation}
    Let $e[s,t]$ be a subedge of an edge of $S$ and let $P$ be given. Then 
    \begin{align*}
        \Cov_{A_S}(e[s,t])&=\bigcup_{(i,j)\in\combinatorial(e[s,t])}[l_i(s),r_j(t)]\\
        &= \left(\bigcup_{(i,j)\in\globalG(e[s,t])}[l_i(s),r_j(t)]\right)\cup \left(\bigcup_{(i,i)\in\localG(e[s,t])}[l_i(s),r_i(t)]\right).
    \end{align*}
\end{observation}

The overall goal is to find a representation of $\Cov_{A_S}(e[s,t])$ that is computationally easy to maintain. Namely, we would like to represent $\Cov_{A_S}(e[s,t])$ as a disjoint union of sets that are easy to maintain.

Let an edge $e$ of $S$ together with $0\leq s\leq t\leq 1$ be given. We say an index $i$ is bad for $e[s,t]$, if all topmost points in cell $i$ of the free space of $e$ and $P$ are to the left of both $l_i(s)$ and $r_i(t)$, and both $l_i(s)$ and $r_i(t)$ are to the left of all bottom most points in cell $i$. Call this set of bad indices $\bad(e[s,t])$. If $i\not\in\bad(e[s,t])$, $i$ is said to be good for $e[s,t]$.

\begin{observation}\label{obs:problematicnice}
    Let $i$ be good for $e[s,t]$ and $P$. If $l_i(s)\neq\infty\neq r_i(t)$ then $l_i(s)\leq r_i(t)$.
\end{observation}

\begin{definition}[Reduced global group]
    Let $e$ be an edge of $S$ and let $0\leq s\leq t\leq 1$ be given. Based on the global group $\globalG(e[s,t])$ we define the reduced global group $\redGlobG(e[s,t])$ which results from $\globalG(e[s,t])$ after merging all pairs of index pairs $(a,b)$ and $(c,d)$ if either $a<c<b<d$, or $b=c$ and $b=c$ is good for $e$.
\end{definition}

\begin{definition}[Proxy coverage]\label{def:proxycov}
For edge $e$ of $S$ and $0\leq s\leq t\leq 1$ define
\begin{multicols}{2}
\vspace*{-3em}
  \begin{equation*}
    \hat{l}_{i,e[s,t]}(s)= 
\begin{cases}
    l_i(s),& \text{if $i$ is good for $e[s,t]$}\\
    r_i(s),              & \text{if $i$ is bad for $e[s,t]$}
\end{cases}
  \end{equation*}\nolinenumbers\break
\vspace*{-3em}
  \begin{equation*}
    \hat{r}_{j,e[s,t]}(t)= 
\begin{cases}
    l_j(t),& \text{if $j$ is good for $e[s,t]$}\\
    r_j(t),              & \text{if $j$ is bad for $e[s,t]$.}
\end{cases}
  \end{equation*}
\end{multicols}
With these at hand, define the \emph{proxy coverage} of subedges of $S$ as the union
\[\proxyCov_{A_S}(e[s,t])=\left(\bigcup_{i \in\localG(e[s,t])\setminus\bad(e[s,t])}[l_i(s),r_i(t)]\right)\cup\left(\bigcup_{(i,j)\in\redGlobG(e[s,t])}[\hat{l}_{i,e[s,t]}(s),\hat{r}_{j,e[s,t]}(t)]\right),\]
where by a slight abuse of notation let $\localG(e[s,t])\setminus\bad(e[s,t])$ be all index-pairs $(i,i)$ in $\localG(e[s,t])$ such that $i$ is not in $\bad(e[s,t])$.
\end{definition}

\begin{figure}
    \centering
    \includegraphics[width=\linewidth]{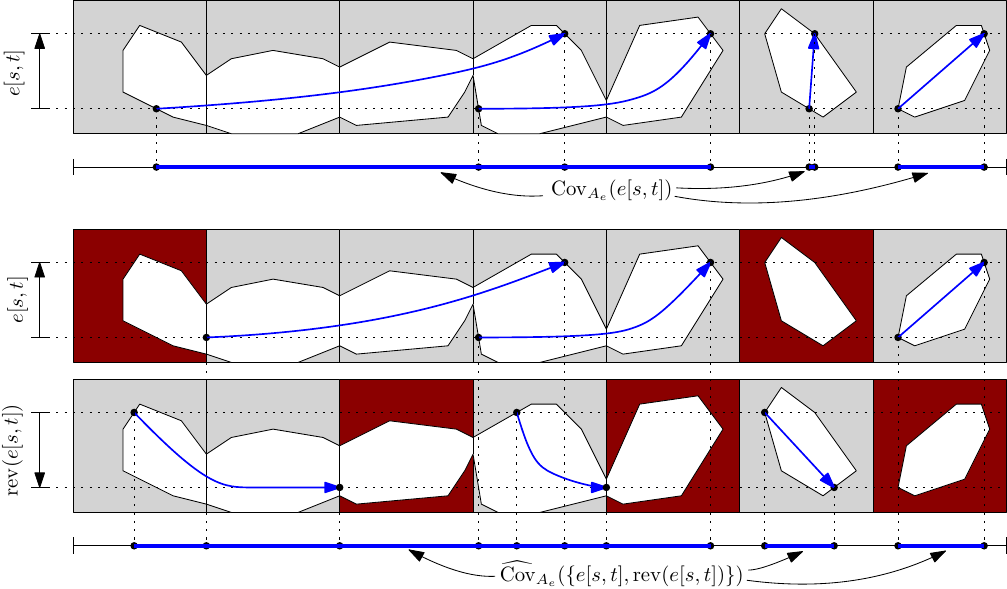}
    \caption{Illustration of the proxy coverage of $e[s,t]$ and $\rev{e[s,t]}$ compared to the coverage of $e[s,t]$. Cells with bad index are marked red. The global group of $e[s,t]$ is $\{(1,4),(4,5)\}$, and the reduced global group of $e[s,t]$ is $\{(1,5)\}$.}
    \label{fig:proxy-coverage}
\end{figure}

We observe that the proxy coverage can in fact be expressed as a disjoint union via the reduced global group.

\begin{lemma} Let $e$ be an edge and let $0\leq s\leq t\leq 1$. Then $\proxyCov_{A_S}(e[s,t])$ coincides with the following disjoint union
\begin{align*}
    \proxyCov_{A_S}(e[s,t])&=\left(\bigsqcup_{i \in\localG(e[s,t])\setminus\bad(e[s,t])}[l_i(s),r_i(t)]\right)\sqcup\left(\bigsqcup_{(i,j)\in\redGlobG(e[s,t])}[\hat{l}_{i,e[s,t]}(s),\hat{r}_{j,e[s,t]}(t)]\right)\\
    &\subset \Cov_{A_S}(e[s,t]).
\end{align*}
\end{lemma}
\begin{proof}
    The only way for two index-pairs $(a,b)$ and $(c,d)$ among the index-pairs in $\localG(e[s,t])\setminus\bad(e[s,t])$ and $\redGlobG(e[s,t])$ to overlap is if $b=c$ is bad for $e[s,t]$ and $(a,b)$ and $(c,d)$ are in $\redGlobG(e[s,t])$. But then $[\hat{l}_{a,e[s,t]}(s),\hat{r}_{b,e[s,t]}(t)]=[\hat{l}_{a,e[s,t]}(s),l_b(t)]$ and $[\hat{l}_{c,e[s,t]}(s),\hat{r}_{d,e[s,t]}(t)]=[r_b(s),\hat{r}_{d,e[s,t]}(t)]$. As $(a,b)$ is in the reduced global group, there is a path which starts to the left of cell $b$ and ends in cell $b$, which implies that the approximate free space in cell $b$ intersects the left boundary and $l_b(s)$ lies to the left of all top-most points in cell $b$. Similarly $r_b(t)$ lies to the right of all bottom-most points. As $b$ is bad, this implies that $l_b(t)<r_b(s)$ and hence $[\hat{l}_{a,e[s,t]}(s),\hat{r}_{b,e[s,t]}(t)]$ and $[\hat{l}_{c,e[s,t]}(s),\hat{r}_{d,e[s,t]}(t)]$ do not intersect, implying the claim that $\proxyCov_{A_S}(\cdot)$ can be written as a disjoint union.\\
    The fact that $l_i(s)\leq \hat{l}_{i,e[s,t]}(s)$ and $\hat{r}_{i,e[s,t]}(t)\leq r_i(t)$ implies the claim that $\proxyCov_{A_S}(\cdot)\subset \Cov_{A_S}(\cdot)$.
\end{proof}

Next we show that $\proxyCov_{A_S}$ approximates $\Cov_{A_S}$ in the sense that for every element $\pi\in \ConcreteCandidates{S}(P)$ there are two elements $\pi_1,\pi_2\in\ConcreteCandidates{S}(P)$ such that 
\[\Cov_{A_S}(\pi)\subset \proxyCov_{A_S}(\pi_1)\cup \proxyCov_{A_S}(\pi_2). \]

We extend the proxy coverage $\proxyCov$ to also be defined for vertex-vertex-subcurves of $S$, where we simply set $\proxyCov_{A_S}(S[s,t])=\Cov_P(S[s,t],\Delta)$. In the remained of this section we will analyze the proxy coverage for subedges only, as the above statement is trivial for vertex-vertex-subcurves by setting $\pi_1=\pi$.

\begin{lemma}\label{lem:revprob}
    If $i$ is bad for $e[s,t]$, then $i$ is good for $\rev{e[s,t]}$. Further $l_i(s),r_i(t)\in[l_i(t),r_i(s)]$.
\end{lemma}
\begin{proof}  
    Let $i$ be bad for $e[s,t]$. This implies that all bottom-most points in cell $i$ of $A_e$ are to the right of all top-most points. As $A_{\rev{e}}$ results from $A_e$ by mirroring it along the $y$-axis, all bottom-most points of cell $i$ in $A_{\rev{e}}$ lie to the left of all top-most points. Hence $i$ is can not be bad for $\rev{e[s,t]}$.

    For the second claim first observe that by \Cref{obs:problematicnice} $l_i(t)<r_i(s)$. Further observe that $t$ lies between $s$ and the $y$-coordinate of the top-most point of cell $i$. Thus by convexity of the free space and the fact that $l_i(s)$ lies to the right of the top-most point there is a point in the free space at height $t$ that is left of $l_i(s)$ and thus in particular $l_i(t)<l_i(s)$. Similarly it follows that $r_i(t)<r_i(s)$ and thus the claim follows.
\end{proof}

\begin{figure}
    \centering
    \includegraphics{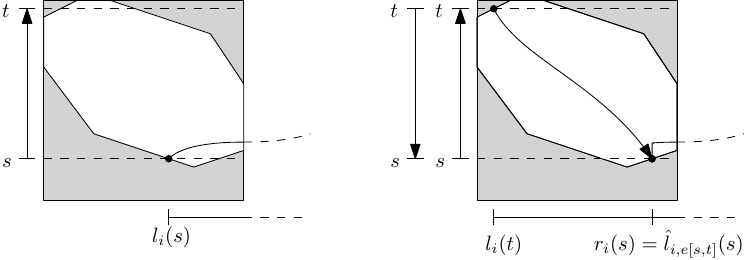}
    \caption{Illustration of how the path starting in cell $i$ from the global group of the coverage of $e[s,t]$ is dominated by the proxy coverage of $e[s,t]$ and $\rev{e[s,t]}$ if $i$ is bad for $e[s,t]$ and hence good for $\rev{e[s,t]}$. Importantly $[l_i(s),r_j(t)]\subset[l_i(t),r_i(s)]\cup[\hat{l}_i(s),r_j(t)]$.
    }
    \label{fig:approx-cov}
\end{figure}

\begin{lemma}\label{lem:proxyapprox}
    The proxy coverage approximates the coverage. That is
    \[\Cov_{A_S}(e[s,t])\subset\proxyCov_{A_S}(e[s,t])\cup \proxyCov_{A_S}(\rev{e[s,t]}).\]
\end{lemma}
\begin{proof}
    If $i$ is in $\localG(e[s,t])\setminus\bad(e[s,t])$, then $[l_i(s),r_i(t)]\subset \proxyCov_{A_S}(e[s,t])$. If instead $i$ is in $\localG(e[s,t])\cap\bad(e[s,t])$ then by \Cref{lem:revprob} it follows that $[l_i(s),r_i(t)]\subset [l_i(t),r_i(s)]\subset \proxyCov_{A_S}(\rev{e[s,t]})$. 

    Let now instead $(i,j)\in\redGlobG(e[s,t])$. We show that $[l_i(s),r_i(s)]\cup[r_i(s),l_j(t)]\cup[l_j(t),r_j(t)]\subset \proxyCov_{A_S}(e[s,t])\cup \proxyCov_{A_S}(\rev{e[s,t]})$. First observe that $[r_i(s),l_j(t)]\subset \proxyCov_{A_S}(e[s,t])$. Assume $i$ is bad, as otherwise $[l_i(s),r_i(s)]\subset [l_i(s),l_j(t)]\subset \proxyCov_{A_S}(e[s,t])$. But then by \Cref{lem:revprob} $[l_i(s),r_i(s)]\subset [l_i(t),r_i(s)]\subset \proxyCov_{A_S}(\rev{e[s,t]})$ (refer to \Cref{fig:approx-cov}). Similarly $[l_j(t),r_j(t)]\subset \proxyCov_{A_S}(e[s,t])\cup \proxyCov_{A_S}(\rev{e[s,t]})$ implying the claim.
\end{proof}

\begin{theorem}\label{thm:16setsystem}
    Let $P$ be a polygonal curve, and let $S$ be a simplification of $P$. 
    Then for any curve $\pi\in\curvespace{\ell}$ there is a set $S_\pi\subset \ConcreteCandidates{S}(P)$ of size at most $16$ such that 
    \[\Cov_P(\pi,\Delta)\subset\bigcup_{s\in S_\pi}\proxyCov_{P}(s).\]
\end{theorem}
\begin{proof}
    This follows from the definition of the proxy coverage, \Cref{lem:proxyapprox} and \Cref{lem:supersetfreespace}.
\end{proof}

\begin{algorithm}
\caption{Maintenance of the reduced global group}\label{alg:main}
\begin{algorithmic}[1]
\Procedure{Maintain}{$e$, $A_e$, sweep-sequence $\sweepseq=\{(s_1,t_1),\ldots,(s_m,t_m)\}$ of $\extremal{}(A_e)$}
    \State for every $i$ compute $B_i=\{(s,t)\in S\mid \text{$i$ is bad for $e[s,t]$}\}$ \Comment{\Cref{lem:badsets}}
    \State Compute coverage $\Cov_{A_e}(e[s_m,t_m])$ represented as $\O(n)$ disjoint intervals
    \State Compute $\localG(e[s_m,t_m])$ and $\globalG(e[s_m,t_m])$ from $\Cov_{A_e}(e[s_m,t_m])$ 
    \State $\unifyG,\localG,\globalG,(s,t)\gets \localG(e[s_m,t_m])\cup \globalG(e[s_m,t_m]),\localG(e[s_m,t_m]),\globalG(e[s_m,t_m]),(s_m,t_m)$
    \State Compute $\bad(e[s,t])$ from $B_i$ and with it $\redGlobG\gets \redGlobG(e[s_m,t_m])$
    \For{$(s',t')$ in \textsc{Reversed}(S)}\label{line:forloop}\Comment{sweeping-window events}
        \If{ $s'\neq s$}
            \State  Compute $\bad(e[s',t])\setminus\bad(e[s,t])$ and $\bad(e[s,t])\setminus\bad(e[s',t])$ via the intervals $B_i$
            \State Update $\unifyG, \localG, \globalG$ and $\redGlobG$ via $\textsc{AdvanceStart}(s,s',t)$
            \State Update $\bad(e[s,t])$ to $\bad(e[s',t])$ via the intervals $B_i$
            \State $s\gets s'$
        \EndIf
        \If{$t'\neq t$}
            \State  Compute $\bad(e[s,t'])\setminus\bad(e[s,t])$ and $\bad(e[s,t])\setminus\bad(e[s,t'])$ via the intervals $B_i$
            \State Update $\unifyG, \localG, \globalG$ and $ \redGlobG$ via $\textsc{AdvanceEnd}(s,t,t')$
            \State Update $\bad(e[s,t])$ to $\bad(e[s',t])$ via the intervals $B_i$
            \State $t\gets t'$ 
        \EndIf
    \EndFor
\EndProcedure
\end{algorithmic}
\end{algorithm}

\subsection{Maintaining the proxy coverage along a sweep-sequence}\label{sec:symbolicRepresentation}

In this section we show that during a sweep-sequence one can correctly maintain a symbolic representation of $\proxyCov_{A_S}(e[s,t])$ in (roughly) logarithmic time per element in the sweep-sequence. To this end we show that one can correctly maintain $\redGlobG(e[s,t])$, $\localG(e[s,t])\setminus\bad(e[s,t])$ and $\bad(e[s,t])$. For this we define the set $\unifyG(e[s,t])$ of inclusion-wise maximal index-pairs $(i,j)$ such that either $i=j$ and $l_i(s)\neq\infty$ or there is a path from cell $i$ at height $s$ to cell $j$ at some height $\hat{t}\leq t$, which helps us maintain $\redGlobG(e[s,t])$, $\localG(e[s,t])\setminus\bad(e[s,t])$.

\begin{lemma}\label{lem:badsets}
    Let $\sweepseq\in\edgeseqs{e}$ be a sweep-sequence. The set $B_i=\{(s,t)\in \sweepseq|\text{$i$ is bad for $e[s,t]$}\}$ is a contiguous subset of $\sweepseq$. They can be computed in total time $\O(n\log n)$
\end{lemma}
\begin{proof}
    This is an immediate consequence of the fact that for the right-most top-most point at $x$-coordinate $x_\mathrm{top}$ the set $\{(s,t)\in \sweepseq|l_i(s)\neq\infty,\ l_i(s)\geq x_\mathrm{top}\}$ is contiguous. So is the set $\{(s,t)\in \sweepseq|r_i(t)\neq\infty,\ r_i(t)\leq x_\mathrm{bottom}\}$, where $x_\mathrm{bottom}$ is the $x$-coordinate of the left-most bottom-most point. $B_i$ is simply the intersection of these two sets. The boundaries of this contiguous set can be computed via binary search over $\sweepseq$.
\end{proof}

Next we show how to maintain $\redGlobG(e[s,t])$. The algorithm is described in \Cref{alg:main}. This algorithm involves maintaining $\globalG(e[s,t])$ which in turn involves maintaining a set $\unifyG(e[s,t])$ consisting of index-pairs $(i,j)$ representing either inclusionwise-maximal paths that end at some height less than $t$, or if $i=j$. To maintain these sets we introduce two data-structures.

\begin{lemma}[\textsc{AdvanceStart} of $\globalG$]\label{lem:advanceStartG}
    Let $s'\leq s\leq t\in \extremal{}(A_e)$ be given. with $s'$ and $s$ consecutive. Let $(i,j)\in\globalG(e[s',t])$. Then one of the following hold:
    \begin{compactenum}
        \item $(i,j)\in\globalG(e[s,t])$,
        \item $(i',j)\in\globalG(e[s,t])$, $s'$ corresponds to the top-most point of the left free space interval of cell $i'$, and $i$ is the smallest index such that $l_{\hat{i}}(s)$ lies in the free space interval of cell $\hat{i}$ for all $i<\hat{i}\leq i'$, or
        \item $(i',j)\in\globalG(e[s,t])$, $s$ corresponds to the lowest point in cell $i'$ and $i$ is the smallest index $i'<i$ such that $l_i(s)\neq\infty$.
    \end{compactenum}
\end{lemma}
\begin{proof}
    Let $(i,j)\in\globalG(e[s',t])$ but $(i,j)\not\in\globalG(e[s,t])$. 
    
    Assume first that there is a path from cell $i$ at height $s$ to cell $j$ at height $t$. This implies that there is an index-pair $(i',j')$ in $\globalG(e[s,t])$ with $i'\leq i< j\leq j'$. Observe that $j'=j$, otherwise there is a path from cell $i$ at height $s'$ to cell $j'$ at height $t$ contradicting the fact that $(i,j)\in\globalG(e[s',t])$. Further observe, that $s'$ must correspond to the lowest point of the cell $i'$ as otherwise $(i',j)\in\globalG(e[s',t])$, contradicting the fact that $(i,j)\in\globalG(e[s',t])$. But then $i$ coincides with the smallest index $i'<\hat{i}$ such that $l_i(s)\neq\infty$, as otherwise $(i,j)$ is dominated by $(\hat{i},j)$.

    Now assume that there is no path from cell $i$ at height $s$ to cell $j$ at height $t$. This in turn implies that there is also no $(i',j')\in\globalG(e[s,t])$ with $i'\leq i < j\leq j'$. But then $s$ must correspond to some top-most point in a free space interval of some cell $i<i'<j$, as otherwise $(i,j)\in\globalG(e[s,t])$. Further clearly $(i',j)\in \globalG(e[s,t])$. This concludes the proof.
\end{proof}

\begin{algorithm}[htp]
\caption{Advancing the start point of the sweep-window}\label{alg:advstart}
\begin{algorithmic}[1]
\Procedure{AdvanceStart}{$s$,$s'$,$t$}
    \State Get $A_e$, $\unifyG$, $\localG$, $\globalG$, $\redGlobG$, $\bad(e[s,t])$, $\bad(e[s',t])\setminus\bad(e[s,t])$ and $\bad(e[s,t])\setminus\bad(e[s',t])$ from \textsc{Maintain}
    \If{$s$ is induced by a lowest point in cell $i$ of $A_e$}
        \State Identify smallest $j>i$ such that $l_j(s)\neq \infty$ \Comment{\Cref{lem:jumpright}}
        \State Update index-pairs in $\unifyG$ and in $\globalG$ starting at $i$ to start at $j$
    \EndIf
    \If{$s'$ is induced by a upper boundary of free space interval of cell $i$ of $A_e$}
        \State Identify largest $j<i$ such that $l_j(s)$ is in the interior of cell $j$ \Comment{\Cref{lem:shootleft}}
        \State Update index-pairs in $\unifyG$ and in $\globalG$ starting at $i$ to start at $j$
    \EndIf
    \If{$s'$ is induced by a highest point in cell $i$ of $A_e$}
        \State Add $(i,i)$ to $\unifyG$
    \EndIf
    \For{changed index-pair $(i,j)$ in $\unifyG$ and $\globalG$}
        \State If $i\neq j$ and $r_j(t)\neq\infty$ move index-pair to $\globalG$, else move it to $\unifyG$
        \State Remove any dominated index-pairs involving $(i,j)$ from $\globalG$ and $\unifyG$
    \EndFor
    \State Update $\redGlobG$ from $\globalG$ after merging overlapping intervals via $\bad(e[s,t])$ and $\bad(e[s',t])$
    \State Update $\localG$ via the changes to $\unifyG$, $\bad(e[s,t])$ and $\bad(e[s,t'])$ \Comment{\Cref{obs:UtoL}}
\EndProcedure
\end{algorithmic}
\end{algorithm}

\begin{lemma}[\textsc{AdvanceStart} of $\unifyG$]\label{lem:advanceStartU}
    Let $s'\leq s\leq t\in \extremal{}(A_e)$ be given. with $s'$ and $s$ consecutive. Let $(i,j)\in\unifyG(e[s',t])$. Then one of the following hold:
    \begin{compactenum}
        \item $(i,j)\in\unifyG(e[s,t])$,
        \item $(i',j)\in\unifyG(e[s,t])$, $s'$ corresponds to the top-most point of the left free space interval of cell $i'$, and $i$ is the smallest index such that $l_{\hat{i}}(s)$ lies in the free space interval of cell $\hat{i}$ for all $i<\hat{i}\leq i'$, or
        \item $(i',j)\in\unifyG(e[s,t])$, $s$ corresponds to the lowest point in cell $i'$ and $i$ is the smallest index $i'<i$ such that $l_i(s)\neq\infty$.
        \item $s'$ corresponds to the highest point in cell $i$ and $i=j$.
    \end{compactenum}
\end{lemma}
\begin{proof}
    Let $(i,j)\in\unifyG(e[s',t])$ but $(i,j)\not\in\unifyG(e[s,t])$. 
    
    Assume first that there is a path from cell $i$ at height $s$ to cell $j$ at height at most $t$. This implies that there is an index-pair $(i',j')$ in $\unifyG(e[s,t])$ with $i'\leq i< j\leq j'$. Observe that $j'=j$, otherwise there is a path from cell $i$ at height $s'$ to cell $j'$ at height at most $t$ contradicting the fact that $(i,j)\in\unifyG(e[s',t])$. Further observe, that $s'$ must correspond to the lowest point of the cell $i'$ as otherwise $(i',j)\in\unifyG(e[s',t])$, contradicting the fact that $(i,j)\in\unifyG(e[s',t])$. But then $i$ coincides with the smallest index $i'<\hat{i}$ such that $l_i(s)\neq\infty$, as otherwise $(i,j)$ is dominated by $(\hat{i},j)$.

    Now assume that there is no path from cell $i$ at height $s$ to cell $j$ at height at most $t$. If $l_i(s)=\infty$ then as $l_i(s')\neq \infty$ we observe that $i=j$ as by symbolic perturbation every highest-most point is above all right-most points and as such there is no monotone path exiting cell $i$.
    If $l_i(s)\neq\infty$ the fact that there is no path from cell $i$ at height $s$ to cell $j$ at height at most $t$ there is no $(i',j')\in\unifyG(e[s,t])$ with $i'\leq i < j\leq j'$. But then $s$ must correspond to some top-most point in a free space interval of some cell $i<i'<j$, as otherwise $(i,j)\in\unifyG(e[s,t])$. Further clearly $(i',j)\in \unifyG(e[s,t])$. This concludes the proof.
\end{proof}

\begin{algorithm}
\caption{Advancing the end point of the sweep-window}\label{alg:advend}
\begin{algorithmic}[1]
\Procedure{AdvanceEnd}{$s$,$t$,$t'$}
    \State Get $A_e$, $\unifyG$, $\localG$, $\globalG$, $\redGlobG$, $\bad(e[s,t])$, $\bad(e[s,t'])\setminus\bad(e[s,t])$ and $\bad(e[s,t])\setminus\bad(e[s,t'])$ from \textsc{Maintain}
    \If{$t$ is induced by a lower boundary of free space interval of cell $i$ of $A_e$}
        \State Identify smallest $j>i$ such that $l_j(s)\neq \infty$ \Comment{\Cref{lem:jumpright}}
        \State Replace index-pairs $(a,b)$ in $\unifyG$ and in $\globalG$ with $a<i\leq b$ with $(a,i)$ and $(j,b)$
    \EndIf
    \If{$t'$ is induced by highest point in cell $j$ of $A_e$}
        \State Add $(i,j)$ to $\globalG$ where $i$ is the minimal index such that $(i,j)\in\unifyG$\Comment{\Cref{lem:UtoG}}
    \EndIf
    \For{changed index-pair $(i,j)$ in $\unifyG$ and $\globalG$}
        \State If $i\neq j$ and $r_j(t)\neq\infty$ move index-pair to $\globalG$, else move it to $\unifyG$
        \State Remove any dominated index-pairs involving $(i,j)$ from $\globalG$ and $\unifyG$
    \EndFor
    \State Update $\redGlobG$ from $\globalG$ after merging overlapping intervals via $\bad(e[s,t])$ and $\bad(e[s,t'])$
    \State Update $\localG$ via the changes to $\unifyG$, $\bad(e[s,t])$ and $\bad(e[s,t'])$ \Comment{\Cref{obs:UtoL}}
\EndProcedure
\end{algorithmic}
\end{algorithm}

\begin{lemma}[\textsc{AdvanceEnd} of $\globalG$]\label{lem:advanceEndG}
    Let $s\leq t'\leq t\in\extremal{}(A_e)$ be given, with $t'$ and $t$ consecutive. Let $(i,j)\in\globalG(e[s,t'])$. Then one of the following holds:
    \begin{compactenum}
        \item $(i,j)\in\globalG(e[s,t])$,
        \item $t'$ corresponds to the highest point of cell $j$, and $i<j$ is the smallest index such that there is a monotone path from cell $i$ starting at height $s$ ending in cell $j$.
        \item $(i,j')\in\globalG(e[s,t])$ and $j'>j$ and $t$ is the lowest-most point in the left free space interval of cell $j+1$
        \item $(i',j)\in\globalG(e[s,t])$ and $i'<i$ and $t$ is the lowest-most point in cell $\hat{i}$ with $i'<\hat{i}<i$ and $i$ is the smallest index $\hat{i}<i$ such that $l_i(s)\neq\infty$.
    \end{compactenum}
\end{lemma}
\begin{proof}
    Let $(i,j)\in\globalG(e[s,t'])$ but $(i,j)\not\in\globalG(e[s,t])$.

    Assume first that there is a path from cell $i$ at height $s$ to cell $j$ at height $t$. This implies that there is an index-pair $(i',j')$ in $\globalG(e[s,t])$ with $i'\leq i< j\leq j'$. But then $t$ must correspond to the lowest point of the left free space interval of some cell $i'<\hat{i}\leq j'$. Further $\hat{i}\leq i$ and $j+1<\hat{i}$, as otherwise $(i,j)\not\in \globalG(e[s,t'])$. But then $i'=i$ or $j=j'$ as otherwise $(i,j)$ is dominated by $(i',\hat{i}-1)$ or $(\hat{i},j')$. If $i'=i$, then $j=\hat{i}$. If $j'=j$, then $i$ is the smallest index $\hat{i}<i$ such that $l_i(s)\neq\infty$.

    Now assume that there is no path from cell $i$ at height $s$ to cell $j$ at height $t$. But then $t$ corresponds to the highest point of cell $j$, and $i$ corresponds to the smallest index $i$ such that there is a monotone path from cell $i$ at height $s$ into cell $j$ concluding the proof.
\end{proof}

\begin{lemma}[\textsc{AdvanceEnd} of $\unifyG$]\label{lem:advanceEndU}
    Let $s\leq t'\leq t\in \extremal{}(A_e)$ be given. with $t'$ and $t$ consecutive. Let $(i,j)\in\unifyG(e[s,t'])$. Then one of the following hold:
    \begin{compactenum}
        \item $(i,j)\in\unifyG(e[s,t])$,
        \item $(i,j')\in\unifyG(e[s,t])$ and $j'>j$ and $t$ is the lowest-most point in the left free space interval of cell $j+1$
        \item $(i',j)\in\unifyG(e[s,t])$ and $i'<i$ and $t$ is the lowest-most point in left free space interval of cell $\hat{i}$ with $i'<\hat{i}<i$ and $i$ is the smallest index $\hat{i}<i$ such that $l_i(s)\neq\infty$.
    \end{compactenum}
\end{lemma}
\begin{proof}
    Let $(i,j)\in\unifyG(e[s,t'])$ but $(i,j)\not\in\unifyG(e[s,t])$.

    Observe that as $t'<t$ the fact that $(i,j)\in\unifyG(e[s,t'])$ implies that there is a path from cell $i$ at height $s$ to cell $j$ at height at most $t$. As $(i,j)\not\in\unifyG(e[s,t])$ there is an index-pair $(i',j')$ in $\unifyG(e[s,t])$ with $i'\leq i< j\leq j'$. But then $t$ must correspond to the lowest point of the left free space interval of some cell $i'<\hat{i}\leq j'$. Further $\hat{i}\leq i$ and $j+1<\hat{i}$, as otherwise $(i,j)\not\in \unifyG(e[s,t'])$. But then $i'=i$ or $j=j'$ as otherwise $(i,j)$ is dominated by $(i',\hat{i}-1)$ or $(\hat{i},j')$. If $i'=i$, then $j=\hat{i}$. If $j'=j$, then $i$ is the smallest index $\hat{i}<i$ such that $l_i(s)\neq\infty$.
\end{proof}

\begin{lemma}\label{lem:symdiff}
    Let $s'\leq s\leq t'\leq t\in \extremal{}(A_e)$ be given. with $s'$ and $s$ consecutive and $t'$ and $t$ consecutive. Then for the symmetric difference $\triangle$ it holds that
    \begin{compactenum}
        \item $|\globalG(e[s,t])\ \triangle\ \globalG(e[s,t'])|=\O(1)$,
        \item $|\globalG(e[s,t])\ \triangle\ \globalG(e[s',t])|=\O(1)$,
        \item $|\unifyG(e[s,t])\ \triangle\ \unifyG(e[s,t'])|=\O(1)$ and
        \item $|\unifyG(e[s,t])\ \triangle\ \unifyG(e[s',t])|=\O(1)$.
    \end{compactenum}
\end{lemma}
\begin{proof}
    It suffices to show that if $t$ is the lowest-most point in cell $\hat{i}$, then there is only one index-pair $(i,j)\in\unifyG(e[s,t])$ such that $i<\hat{i}\leq j$. This however is an immediate consequence of the fact that if $i<\hat{i}\leq j$ and $(i',j')\in\unifyG(e[s,t])$ with $i'<\hat{i}\leq j'$, then both paths go through the bottom-most point of the free space interval of cell $\hat{i}$, and thus $j=j'$. Further $i$ is the minimal index such that there is a monotone path to the bottom-most point of the free space interval of cell $\hat{i}$, and so is $i'$, and thus $i=i'$. This argument also carries over for $\globalG$.
    
    Lastly, by inclusion-wise maximality of index-pairs in $\globalG(e[s,t])$ and $\unifyG(e[s,t])$ any index $i$ occurs at most $\O(1)$ times among all index-pairs in $\globalG(e[s,t])$ and $\unifyG(e[s,t])$. Then \Cref{lem:advanceStartG}, \Cref{lem:advanceStartU}, \Cref{lem:advanceEndG} and \Cref{lem:advanceEndU} conclude the proof.
\end{proof}

\begin{lemma}\label{lem:UtoG}
    Given $\unifyG(e[s,t])$ and index $j$, one can compute the minimal index $i<j$, such that there is a path from cell $i$ at height $s$ to cell $j$ at height $t$ or decide that this index does not exist in time $\O(\log n)$.
\end{lemma}
\begin{proof}
    The index $i$ is precisely the smallest index $i$ such that $(i,j')\in\unifyG(e[s,t])$ such that $i<j\leq j'$. This index $i$ can be computed via binary search over the index-pairs stored in $\unifyG(e[s,t])$.
\end{proof}

 \begin{observation}\label{obs:UtoL}
     $\localG(e[s,t])=\{(i,j)\in\unifyG(e[s,t])\mid i=j \text{, $r_j(t)\neq\infty$ and } i\not\in\bad(e[s,t])\}$.
 \end{observation}

\begin{lemma}[Shoot-left data-structure]\label{lem:shootleft}
    Let $I=\{i_1,\ldots,i_n\}$ be a list of intervals in $\bR$. One can build in $O(n\log n)$ time a data-structure that allows queries with input $i\leq n$ and $x\in\bR$ correctly outputting the smallest index $j\leq i$ such that $x\in\bigcap_{j\leq s\leq i}i_s$ or determining that there is no such $j$.
\end{lemma}
\begin{proof}
    Store the intervals in a balanced binary tree, where every parent stores the intersection of the intervals stored in its children nodes. This construction takes $\O(n\log n)$ time. Thus if the query point $x$ lies in the interval stored at some node $r$ of the tree, then it also lies in the intersection of all intervals stored in the leaves of the tree rooted at $r$.

    Let now $i$ and $x$ be given. If $x\not \in i_i$, return that there is not such $j$. Otherwise we temporarily modify the tree as follows: Traverse the tree upwards starting at the node of $i_i$ removing all children of nodes that represent trees of intervals whose index is strictly bigger than $i$, updating the intervals along the way. This modification takes $\O(\log n)$ time and results in a tree of height $\log n$ whose leaves are all intervals $\{i_1,\ldots,i_i\}$ such that every every node stores the intersection of all intervals of the leaves on its subtree. Next we traverse the tree starting at the root, checking whether $x$ lies in the stored interval at the root. If it we can output $1$. Otherwise assume there is an interval $i_j$ which is the interval with the largest index that does not contain $x$. We now recursively check both children. If the interval of the right child (towards $i$) does not contain $x$, we next traverse in this child. If instead the interval of the right child does not contain $x$, then the left child must not contain $x$, hence we traverse right. In the end we correctly identified the leaf $j$.
\end{proof}

\begin{lemma}[Jump-right data-structure]\label{lem:jumpright}
    Let $A=\{a_1,\ldots,a_n\}$ be a list of values. One can build in $\O(n\log n)$ time a data-structure that allows queries with input $i$ correctly outputting the smallest index $j>i$ such that $a_i>a_j$ or determining that there is no such $j$ in $\O(1)$ time.
\end{lemma}
\begin{proof}
    First sort the indices $\{1,\ldots,n\}$ with the values $\{a_1,\ldots,a_n\}$ as keys in time $\O(n\log n)$. The resulting list of indices is then iteratively inserted into a sorted list (this time by index). Each insertion takes $\O(\log n)$ time. Further, when some index $i$ is inserted at position $k$, the smallest index $j>i$ such that $a_i>a_j$ (if it exists) is at position $k+1$. Hence, when inserting this index $i$, look up the index at position $k+1$ and store it as the answer to the query with index $i$.
\end{proof}

\begin{figure}
    \centering
    \includegraphics[width=\linewidth]{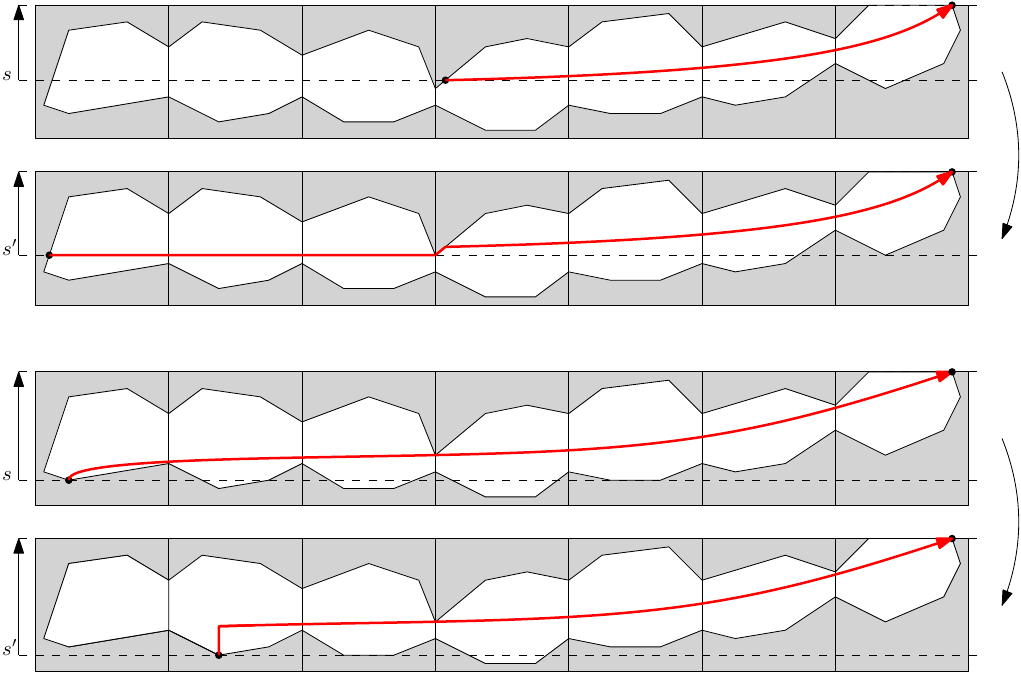}
    \caption{Shoot-left, and jump-right update to $\globalG(e[s,t])$.}
    \label{fig:shoot-left}
\end{figure}

\begin{theorem}\label{thm:maintainSymbolic}
    At the beginning and end of the loop in Line $7$ of \textsc{Maintain} of \Cref{alg:main}, the set of index-pairs $\redGlobG$ coincides with $\redGlobG(e[s,t])$ and $\localG$ coincides with $\localG(e[s,t])\setminus\bad(e[s,t])$. Further, executing \textsc{Maintain} takes $\O(n\log n)$ time.
\end{theorem}
\begin{proof}

    Observe that initially $\bad(e[s,t])=\emptyset$ and $\unifyG(e[s,t])=\unifyG(e[s,s])=\globalG(e[s,s])\cup\localG(e[s,t])$ which is correctly computed.

    Observe that by \Cref{lem:advanceStartG}, \Cref{lem:advanceStartU}, \Cref{lem:advanceEndG}, \Cref{lem:advanceEndU}, \Cref{lem:UtoG} and \Cref{obs:UtoL} the sets $\unifyG=\unifyG(e[s,t])$, $\localG=\localG(e[s,t])\setminus\bad(e[s,t])$ and $\globalG = \globalG(e[s,t])$ are correctly updated by \textsc{AdvanceStart} and \textsc{AdvanceEnd}.
    
    By \Cref{lem:symdiff}, throughout the entirety of the algorithm there are $\O(n)$ updates to $\globalG$ that can be described as a total of $\O(n)$ insertions and removals of index-pairs, where at any point no index-pair is contained in another index-pair. Under these assumptions $\redGlobG$ can be maintained as $\bad(e[s,t])$ is given, and at every (sorted) insertion and removal only its two neighboring index-pairs need to be considered for additional updates to $\redGlobG$. Maintaining $\redGlobG$ thus takes total time $\O(n\log n)$.
\end{proof}

\begin{corollary}\label{cor:intervalset}
    Let $\sweepseq\in\edgeseqs{e}$ be a sweep-sequence. There are $m=\O(n)$ index-pairs $p_1=(i_1,j_1),\ldots,p_m=(i_m,j_m)$ together with $m$ contiguous subsets $I_i=\{(s_{a_i},t_{a_i}),\ldots,(s_{b_i},t_{b_i})\}\subset S$, such that for every $(s,t)\in \sweepseq$ it holds that
    \[\redGlobG(e[s,t])\cup(\localG(e[s,t])\setminus\bad(e[s,t]))=\bigcup_{1\leq i\leq m, (s,t)\in I_i}\{p_i\}.\]
    Further for $p_k=(i_k,j_k)$ the index $i_k$ is either always good or always bad for $e[s,t]$ for $(s,t)\in I_k$, and the index $j_k$ is either always good or always bad for $e[s,t]$ for $(s,t)\in I_k$.
    The index-pair $p_i$ as well as the values $a_i$ and $b_i$ can be computed for all $i$ in total time $\O(n\log n)$.
\end{corollary}
\begin{proof}
    This is an immediate consequence of \Cref{thm:maintainSymbolic}, as throughout \Cref{alg:main} $\redGlobG(e[s,t])\cup(\localG(e[s,t])\setminus\bad(e[s,t]))$ is correctly maintained along $\sweepseq$ in $\O(n)$ updates to $\O(n)$ initial index-pairs. We store any index-pair in $\redGlobG(e[s,t])\cup(\localG(e[s,t])\setminus\bad(e[s,t]))$ whenever it gets added, removed, modified or either of its entries changes from good to bad proving the claim.
\end{proof}

Thus, overall, we are able to maintain a symbolic representation of $\proxyCov(\cdot)$ during a sweep-sequence in total time $\O(\log n)$.

Next, we show how to use this maintained set $\redGlobG(\cdot)$ and $\localG(\cdot)\setminus\bad(\cdot)$ for batch point-queries.

\begin{theorem}\label{thm:pointQuery}
    Let $\sweepseq\in\edgeseqs{e}$, and let $Q\subset[0,1]$ together with $w_Q:Q\rightarrow\bN$ be a weighted point-set. Let for any $Q'\subset Q$ its weight $w_Q(Q')$ be defined as $\sum_{q\in Q'}w_q(q)$. There is an algorithm which computes $w_Q(Q\cap \proxyCov_{A_S}(e[s,t]))$
    for every $(s,t)\in \sweepseq$ in $\O(|Q|\log|Q| + n\log n)$ time.
\end{theorem}
\begin{proof}
    Define following values for all $i$, $j$ along the sweep-sequence $\mathfrak{s}$ where for $(s,t)\in\sweepseq$:
    \begin{multicols}{2}
\vspace*{-3em}
  \begin{equation*}
    L_i((s,t))=\hspace{-2em}\sum_{\substack{q\in Q,\ q\text{ in cell $i$}\\\hat{l}_{i,e[s,t]}(s)\neq\infty,\ \hat{l}_{i,e[s,t]}(s)\leq q}}\hspace{-2em}w_Q(q),
  \end{equation*}\nolinenumbers\break
\vspace*{-3em}
  \begin{equation*}
    R_j((s,t))=\hspace{-2em}\sum_{\substack{q\in Q,\ q\text{ in cell $j$}\\\hat{r}_{j,e[s,t]}(t)\neq\infty,\ q\leq \hat{r}_{j,e[s,t]}(t)}}\hspace{-2em}w_Q(q),
  \end{equation*}
\end{multicols}

\begin{multicols}{2}
\vspace*{-3em}
  \begin{equation*}
    M_i((s,t))=\hspace{-3em}\sum_{\substack{q\in Q,\ i\not\in\bad(e[s,t]),\ q\text{ in cell $i$}\\l_i(s)\neq\infty,\ r_i(t)\neq\infty,\ l_i(s)\leq q\leq r_i(t)}}\hspace{-4em}w_Q(q),
  \end{equation*}\nolinenumbers\break
\vspace*{-3em}
  \begin{equation*}
    D(i,j)=\hspace{-1.5em}\sum_{\substack{q\in Q,\ q\text{ in cell $m$}\\i\leq m\leq j}}\hspace{-1.5em}w_Q(q).
  \end{equation*}
\end{multicols}
The value $L_i(s,t)$ corresponds to the weight of points in cell $i$ that lie right of $\hat{l}_{i,e[s,t]}(s)$. Similarly $R_i(s,t)$ corresponds to the weight of points in cell $i$ that lie left of $\hat{r}_{i,e[s,t]}(t)$. The value $M_i(s,t)$ corresponds to the weight of points in cell $i$ that lie in between $\hat{l}_{i,e[s,t]}(s)$ and $\hat{r}_{i,e[s,t]}(t)$ if $i$ is good for $e[s,t]$. Lastly $D(i,j)$ corresponds to the total weight of points that lie on edge $i,\ldots,j$. Then for any $(s,t)\in\sweepseq$ and any $(i,i)\in\localG(e[s,t])\setminus\bad(e[s,t])$ it holds that \(M_i((s,t)) = w_Q(Q\cap[l_i(s),r_i(t)])\). Similarly, for any $(i,j)\in\redGlobG(e[s,t])$ it holds that \(L_i((s,t)) + D(i+1,j-1) + R_j((s,t))\) equals \(w_Q(Q\cap[\hat{l}_{i,e[s,t]}(s),r_{i,e[s,t]}(t)])\).
Hence
\begin{align*}
    &w_Q(Q\cap\proxyCov_{A_S}(e[s,t]))=\\
    &=\left(\sum_{i \in\localG(e[s,t])\setminus\bad(e[s,t])}\hspace{-1.5em}w_Q(Q\cap[l_i(s),r_i(t)])\right) +\left(\sum_{(i,j)\in\redGlobG(e[s,t])}\hspace{-1em}w_Q(Q\cap[\hat{l}_{i,e[s,t]}(s),\hat{r}_{j,e[s,t]}(t)])\right)\\
    &=\left(\sum_{i \in\localG(e[s,t])\setminus\bad(e[s,t])}\hspace{-1.5em}M_i((s,t))\right)+\left(\sum_{(i,j)\in\redGlobG(e[s,t])}\hspace{-1em}L_i((s,t)) + D(i+1,j-1) + R_j((s,t))\right).
\end{align*}

Observe that $D(i,j)$ can be provided via a data-structure that first computes $d_i=\sum_{q\in Q,q\text{ in cell $i$}}w_Q(q)$ for every $i$ in total time $\O(|Q|\log n)$ and stores them in a balanced binary tree as leaves, where every inner node stores the sum of the values of its children. For every $i\leq j$ the value $D(i,j)=\sum_{i\leq m\leq j}d_m$ can then be returned in $\O(\log n)$ time by identifying in $\O(\log n)$ time all $\O(\log n)$ maximal subtrees whose children lie in the interval $[i,j]$ and then returning the sum of the stored values in the root of each maximal subtree.

\begin{figure}
    \centering
    \includegraphics{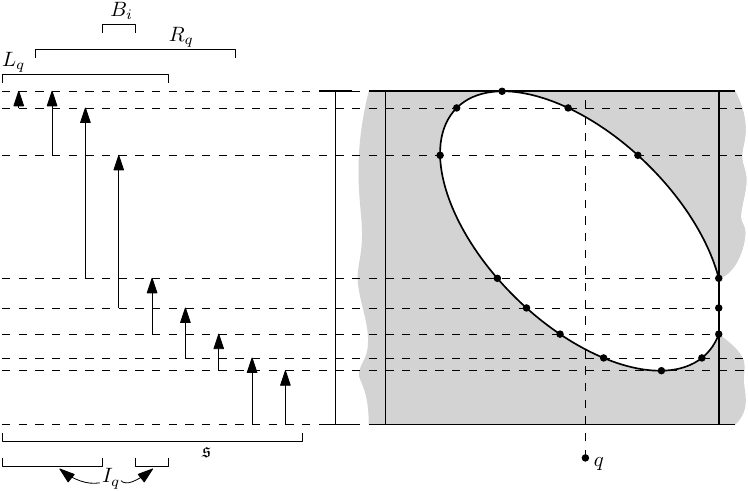}
    \caption{Construction of the set $I_q$ encoding when $\hat{l}_{i,e[s,t]}(s)\leq q$ from Proof of \Cref{thm:pointQuery} via the sets $L_q$ and $R_q$ encoding when $l_i(s)\leq q$ and $r_i(s)\geq q$ and the set $B_i$ for all $(s,t)\in\sweepseq$.}
    \label{fig:Iq}
\end{figure}

Next we show that $L_i(\cdot)$ (resp. $R_j(\cdot)$ and $M_i(\cdot)$) can correctly be maintained when performing the sweep of $\sweepseq$. To this end let \[I_{q}=\{(s,t)\in \sweepseq\mid\text{$q$ is in cell $i$ and } \hat{l}_{i,e[s,t]}(s)\neq\infty,\hat{l}_{i,e[s,t]}\leq q\}.\] 
Refer to \Cref{fig:Iq}. As the free space in every cell is convex, throughout $\sweepseq$ the first indices are monotone and the $y$-coordinates of the left-most and right-most points are stored in $\extremal{}(A_e)$, for every $q\in Q$ the contiguous subsets \[L_{q}=\{(s,t)\in\sweepseq\mid\text{$q$ is in cell $i$ and } l_{i}(s)\neq\infty,l_i(s)\leq q\}\text{ and }\]\[R_{q}=\{(s,t)\in\sweepseq\mid \text{$q$ is in cell $i$ and } l_{i}(s)\neq\infty,r_i(s)\geq q\}\] can be computed in $\O(\log n)$ time. Similarly the set $B_i=\{(s,t)\in\sweepseq\mid \text{$i$ is bad for $e[s,t]$}\}$ is also a contiguous subset of $\sweepseq$ and all $B_i$ can be computed in $\O(\log n)$ time. Let $Q_i$ denote all $q\in Q$ that lie in the $i$th edge of $P$. Then for any $q\in Q_i$
\[I_q = (L_q\cap(\sweepseq\setminus B_i)) \cup ((\sweepseq\setminus R_q) \cap B_i),\]
and thus $I_q$ consists of $\O(1)$ contiguous disjoint subsets of $\sweepseq$. Thus all sets $I_q$ (represented as $\O(1)$ contiguous subsets of $\sweepseq$) can be computed in time $\O(|Q|\log n + n\log n)$. Further,
\[L_i((s,t)) = \sum_{q\in Q_i}\mathbbm{1}_{q\in I_q}w_Q(q),\]
and hence $L_i(\cdot)$ can be maintained in total time $\O(|Q| + n\log n)$ after one initial computation of all $I_{q}$, by adding $w_Q(q)$ for $q\in Q_i$ whenever $(s,t)$ enters the $\O(1)$ contiguous disjoint subsets of $I_q$, and subtracting $w_Q(q)$ for $q\in Q_i$ whenever $(s,t)$ exits the $\O(1)$ contiguous disjoint subsets of $I_q$. Sorting the boundaries of all $I_q$ preparing them for the maintenance of $L_i(\cdot)$ takes $\O(|Q|\log|Q|)$ time. Similarly $M_i(\cdot)$ and $R_j(\cdot)$ can be maintained.

Overall, the values $L_i(\cdot)$, $R_j(\cdot)$, $M_i(\cdot)$ and $D(i,j)$ are correctly maintained in total time $\O(|Q|\log|Q|)$ time such that they can be evaluated in $\O(\log n)$. By \Cref{cor:intervalset} there are only $\O(n)$ total updates to $\redGlobG(\cdot)$ and $\localG(\cdot)\setminus\bad(\cdot)$. Hence, $w(\cdot)$ can be correctly maintained along $\sweepseq$ by updating it whenever $L_i(\cdot)$, $R_j(\cdot)$, $M_i(\cdot)$, $\redGlobG(\cdot)$, $\localG(\cdot)\setminus\bad(\cdot)$ or $\bad(\cdot)$ change. Thus computing $w(e[s,t])$ for all $(s,t)\in S$ takes $\O(|Q|\log|Q| +n\log n)$ time.
\end{proof}

\section{Cubic Subtrajectory Covering}\label{sec:cubicClusterin}

In this section we focus on the SC problem. The goal of this section is the following theorem.

\begin{restatable}{theorem}{thmMainClustering}
\label{thm:mainClustering}
    Let $P$ be a polygonal curve of complexity $n$ and let $\Delta>0$ and $\ell\leq n$ be given. Let $\kopt$ be the size of the smallest set $\mathcal{C}^*\subset\curvespace{\ell}$ such that $\bigcup_{c\in\mathcal{C}^*}\Cov_P(c,\Delta)=[0,1]$. There is an algorithm that given $P$, $\Delta$ and $\ell$ computes a set $\mathcal{C}\subset\curvespace{\ell}$ of size $\O(\kopt\log n)$ such that $\bigcup_{c\in\mathcal{C}}\Cov_P(c,4\Delta)=[0,1]$. Further, it does so in $\O(n^3\log ^2 n + kn^2\log^3 n)$ time using $\O(n^3\log n)$ space.
\end{restatable}

Throughout this section let $S$ be a simplification of $P$ and set $\alpha=1$, that is, the approximate free spaces $A_e$ and $A_S$ are exactly the $4\Delta$-free space of $e$ and $P$, and $S$ and $P$, respectively. We further set $\extremal{}(A_S)$ to be precisely all extremal points of the $4\Delta$-free space of $S$ and $P$.

\subsection{Atomic intervals}
\begin{definition}[Atomic intervals]
    Let $P$ be a polygonal curve, and let $\Delta>0$ and $\ell$ be given. Let $S$ be a simplification of $P$. Let $G$ be the set of all intersection-points of horizontal lines at height $y$ for $y\in\extremal{}(A_S)$ with the boundary of the free space $\mathcal{D}_{4\Delta}(S,P)$. From this, define the set of atomic intervals $\atomic{}(S,P)$ as the set of intervals describing the arrangement of $[0,1]$ defined by the set $\{x\in[0,1]\mid \exists y\in[0,1]: (x,y)\in G\}$ (\Cref{fig:extremal-molecular}).
\end{definition}

\begin{observation}\label{obs:atomicnumber}
    As $|\extremal{}(A_S)|\leq 8n^2$, and each horizontal line intersects at most $n$ cells, and the free space in every cell is convex it follows, the set of all midpoints of atomic intervals $\atomic{}(S,P)$ is a point set $A\subset[0,1]$ of size $16n^3$ such that for any $C\subset\ConcreteCandidates{S}(P)$ it holds that 
    \[\proxyCov_{P}(C) = [0,1]\iff A \subset \proxyCov_{P}(C).\]
\end{observation}

\begin{theorem}\label{thm:reduction}
    Let $P$ be curve of complexity $n$, let $\Delta$ and $\ell$ be given. Let $S$ be a simplification of $P$. Let $A\subset[0,1]$ be a finite set. Any algorithm that iteratively adds the curve $c$ among $\ConcreteCandidates{S}(P)$ to $R$ maximizing
    \[\left|\left\{a\in A\middle|a\in\left(\proxyCov_P(c,4\Delta)\setminus\left(\bigcup_{r\in R}\proxyCov_P(r,4\Delta)\right)\right)\right\}\right|,\]
    terminates after $16(\ln |A| + 1)\kopt$ iterations where $\kopt$ is the size of the smallest set $\mathcal{C}^*\subset\curvespace{\ell}$ such that $\bigcup_{c\in\mathcal{C}^*}\Cov_P(c,\Delta)=[0,1]$. Furthermore for the resulting set $R\subset\bX_\ell^d$ it holds that
    \(\Cov_P(R,4\Delta)=[0,1]\).
\end{theorem}
\begin{proof}
    By \Cref{thm:16setsystem} there is a set $C_1$ of size $16\kopt$ in $\ConcreteCandidates{S}(P)$ such that $\proxyCov_{A_S}(C_1)=[0,1]$. Hence, by standard greedy \textsc{SetCover} arguments, the algorithm thus terminates after $(\ln|A|+1)16\kopt$ iterations, returning a set $C\subset\ConcreteCandidates{S}(P)$ of size $(\ln |A|+1)16\kopt$ such that $\proxyCov_{A_S}(C)=[0,1]$. As $\proxyCov_{A_S}(\cdot)\subset\Cov_P(\cdot,4\Delta)$ it follows that $\Cov_P(C,4\Delta)=[0,1]$.
\end{proof}

We define the weight of any subcurve $S[s,t]$ in $\ConcreteCandidates{S}(P)$ with respect to the point-set $A\subset[0,1]$ as
\(\candW{A}(S[s,t]):=\left|A\cap\proxyCov_{A_S}(S[s,t])\right|\).
By \Cref{thm:reduction} it suffices to show that given a set of curves $R\subset \ConcreteCandidates{S}(P)$ in we can evaluate $\candW{A\setminus \proxyCov_{A_S}(R)}(\pi)$ for every $\pi\in\ConcreteCandidates{S}(P)$. 

\subsection{\boldmath Maintaining a partial solution and Type $(I)$-curves}

The greedy-algorithm operates in rounds: given a partial solution $R\subset\ConcreteCandidates{S}(P)$, which initially is empty, it identifies a subcurve $\pi\in\ConcreteCandidates{S}(P)$ which maximizes $\candW{A}(\pi,R)$. Let $I_R$ be the $\O(|R|n)\leq \O(\kopt n\log n )$ disjoint intervals describing $\bigcup_{r\in R}\Cov_P(r,4\Delta)$. During the greedy set cover algorithm the exact information on $R$ is less important, more so $I_R$, and accessing and updating this union of intervals quickly throughout the algorithm, as by definition
\[\candW{A\setminus\proxyCov_{A_S}(R)}(\pi)=\left|\left\{a\in A\middle|a\subset\left(\proxyCov_{A_S}(\pi)\setminus\left(\bigcup_{i\in I_R}i\right)\right)\right\}\right|.\]

\begin{lemma}\label{lem:solution_datastructure}
    Let $I$ be a set of $n$ pairwise disjoint intervals defined by interval boundaries of $\atomic{}(S,P)$. There is a data-structure that allows computing

\[\left|\left\{a\in A\middle|a\subset\left([l,r]\setminus\left(\bigcup_{i\in I_R}i\right)\right)\right\}\right|\]
    
    in $\O(\log n)$ time for $[l,r]$ defined by boundaries of $\atomic{}(S,P)$. Further we can either remove an interval from $I$, or if $[l,r]\cap \left(\bigcup_{i\in I}i\right)=\emptyset$, then $[l,r]$ can be added to $I$ correctly updating the data-structure $\O(\log n)$ time.
\end{lemma}
\begin{proof}
    Simply store the disjoint intervals $I_R$ in an interval-tree and for every interval store the number of intervals from $\atomic{}(S,P)$ that lie inside it. This number can be computed in $\O(\log n)$ time.
\end{proof}

Throughout the $\O(\kopt \log n)$ iterations of the algorithm $I$ will have complexity at most $\O(n\kopt \log n)$, and maintaining it takes at most $\O(n\kopt \log^2 n)$ time, as any interval is added and removed exactly once to maintain disjointedness. 
\begin{observation}
    Trivially $\kopt \leq\lceil n/\ell\rceil=\O(n)$.
\end{observation}

\begin{lemma}[\cite{vanderhoog2024fasterdeterministicsubtrajectoryclustering}]\label{lem:evalType1}
    In $\O(n^2\ell\log\ell\log n)$ one can compute $\Cov_{P}(S[s,t],4\Delta)=\proxyCov_{A_S}(S[s,t])$ for every Type $(I)$-subcurve $S[s,t]$ of $S$.
\end{lemma}

Given $\O(n)$ disjoint intervals describing some $\proxyCov_{A_S}(\pi)$ and the data structure from \Cref{lem:solution_datastructure} storing a partial solution consisting of at most $\O(n^2\log n)$ disjoint intervals. Then we can compute $\candW{A}(\pi,R)$ in time $\O(n\log n)$. Thus the following theorem follows.

\begin{theorem}\label{thm:type1}
    Let $S$ and $P$ be polygonal curves of complexity $n$. After $\O(n^2\ell\log^2 n)$ preprocessing time one can answer queries with a set $I_R$ of intersection-free intervals together with the data-structure from \Cref{lem:solution_datastructure} representing the coverage of a set $R\subset\ConcreteCandidates{S}(P)$ of size $\O(k\log n)$ identifying the Type $(I)$ subcurve $c$ of $S$ maximizing $\candW{A}(c,I)$ in $\O(n^2\log n\log \ell)$ time.
\end{theorem}
\begin{proof}
    This is an immediate consequence of \Cref{lem:evalType1} and \Cref{lem:solution_datastructure}. We compute and store the $\O(n)$ intervals representing the coverage for every Type $(I)$ subcurve in total time $\O(n^2\ell\log^2n)$.
\end{proof}

\subsection{\boldmath Type $(II)$- and $(III)$-subcurves}

This leaves us with computing the weight of subcurves in the sweep-sequences, both of which are subedges.

\begin{observation}
    For any edge $e$ of $S$ it holds that $\extremal{}(A_e)\subset\extremal{}(A_S)$ and as such the atomic intervals of $e$ and $P$ partition the atomic intervals of $S$ and $P$, where each atomic interval of $e$ and $P$ can be described as a contiguous subset of atomic intervals of $S$ and $P$.
\end{observation}

\begin{figure}
    \centering
    \includegraphics[width=\linewidth]{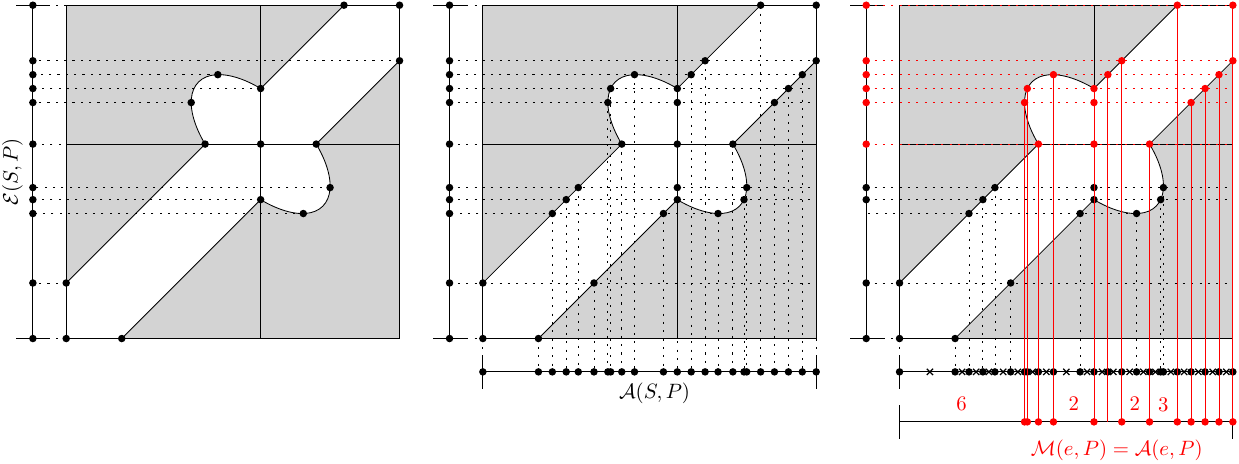}
    \caption{Illustration of extremal points $\extremal{}(S,P)$ in the $4\Delta$-freespace of $S$ and $P$, as well as the atomic intervals $\atomic{}(S,P)$. On the right in red molecular intervals $\atomic{}(e,P)$ for the second edge $e$ of $S$ are illustrated. Their weights w.r.t. the midpoints of atomic intervals $\atomic{}(S,P)$ are shown if they are not equal to $1$.}
    \label{fig:extremal-molecular}
\end{figure}

To clarify the difference of atomic intervals $\atomic{}(S,P)$ and atomic intervals $\atomic{}(e,P)$ we will call the latter \emph{molecular intervals} (\Cref{fig:extremal-molecular}) and denote them by $\molecular{}(e,P)$.

A molecular interval $m\in\molecular{}(e,P)$ is additionally endowed with a weight \(\molW{A}(m)=\left|A\cap m\right|\).

\begin{lemma}\label{lem:molecularPQ}
    Let $A\subset[0,1]$ be a set of $\O(n^3)$ points. Let $e$ be an edge of $S$. Let $W_e=\{(s,t)\in\molecular{}(e,P)\mid w_A(e[s,t])\neq 0\}$ together with $w_A(m)$ and the midpoint of $m$ for every $m\in W_e$ be given. Let $\sweepseq\in\edgeseqs{e}$ be a sweep-sequence. Then $\candW{A}(e[s,t])$ can be computed for every $(s,t)\in\sweepseq$ in total time $\O(|W_e|\log n + n\log n)$.
\end{lemma}
\begin{proof}
    This is an immediate consequence of \Cref{thm:pointQuery} together with the fact that 
    \begin{align*}
        \candW{A}(e[s,t])&=|A\cap\proxyCov_{A_S}(e[s,t])|=\sum_{[a,b]\in\molecular{}(e,P)}\mathbbm{1}_{[a,b]\subset\proxyCov_{A_S}(e[s,t])}|A\cap[a,b]|\\
        &= \sum_{[a,b]\in\molecular{}(e,P)}\mathbbm{1}_{[a,b]\subset\proxyCov_{A_S}(e[s,t])}\molW{A}([a,b]).\qedhere
    \end{align*}
\end{proof}

\begin{lemma}\label{lem:theDestroyer}
    Let $A$ be a set of size $\O(n^3)$ stored in a tree. Let $W_e=\{m\in\molecular{}(e,P)\mid \molW{A}(m)\neq 0\}$ be stored together with their weights $\molW{A}(\cdot)$ in a self-balancing tree where $m$ is left of $m'$ in the tree, if $m$ lies to the left of $m'$ in $[0,1]$. Let $r$ together with its proxy coverage $\proxyCov_{A_S}(r)$ be given. Let $\widetilde{W}_e=\{m\in\molecular{}(e,P)\mid \molW{A\setminus \proxyCov_{A_S}(r)}(m)\neq 0\}\subset W_e$.
    \begin{compactenum}
        \item In time $\O(|W_e\setminus \widetilde{W}_e|+n\log n)$ one can compute and output $W'=\{m\in\molecular{}(e,P)\mid \molW{A\cap \proxyCov_{A_S}(r)}(m)\neq 0\}\subset W_e$ together with its weights $\molW{A\cap \proxyCov_{A_S}(r)}(\cdot)$.
        \item In total time $\O(|W_e\setminus \widetilde{W}_e|+n\log n)$ one can update the tree storing $W_e$ and its weights $\molW{A}(\cdot)$ to instead store $\widetilde{W}_e$ and its weights $\molW{A\setminus \proxyCov_{A_S}(r)}(\cdot)$.
        \item In time $\O(|A\cap\proxyCov_{A_S}(r)|\log n + n\log n)$ one can compute a tree storing $A\cap\proxyCov_{A_S}(r)$.
        \item In time $\O(|A\cap\proxyCov_{A_S}(r)| + n\log n)$ one can update the tree storing $A$ to store $A\setminus\proxyCov_{A_S}(r)$ instead.
    \end{compactenum}
\end{lemma}
\begin{proof}
    Let $[a_1,b_1],\ldots$ be the $\O(n)$ disjoint intervals of atomic intervals representing $\proxyCov_{A_S}(r)$. Let $m\in W_e$. If $m\subset [a_i,b_i]$ then $\molW{A\cap \proxyCov_{A_S}(r)}(m) = \molW{A}(m)$ and hence $m\in W'$ and $m\not\in\widetilde{W}_e$. If $m\subset [0,1]\setminus\bigcup_i[a_i,b_i]$ then $\molW{A\cap \proxyCov_{A_S}(r)}(m) = 0$ and as such $m\not \in W'$ and $m\in\widetilde{W}_e$. Lastly there are $\O(n)$ molecular intervals that contain some interval boundary $a_i$ or $b_i$.
    
    \textbf{1:} To compute and output $W'$ output all $m$ and their weight $\molW{A}(\cdot)$ that lie in some interval $[a_i,b_i]$ in $\O(|W_e\setminus \widetilde{W}_e|+n\log n)$ time via the tree storing $W_e$. Further determine the $\O(n)$ intervals that contain some interval boundary $a_i$ or $b_i$ and for them compute $\molW{A\cap\proxyCov_{A_S}(r)}(\cdot)$ via the tree storing $A$ in total time $\O(n\log n)$.

    \textbf{2:} To update the tree to store $\widetilde{W}_e$ remove all molecular intervals from the tree that are contained in some interval $[a_i,b_i]$ in total time $\O(|W_e\setminus \widetilde{W}_e|+n\log n)$. Lastly for the $\O(n)$ intervals that contain some interval boundary $a_i$ or $b_i$ compute and store their weight $\molW{A\setminus \proxyCov_{A_S}(r)}(\cdot)$. If this weight is zero, also remove them from the tree. Handling these $\O(n)$ intervals takes $\O(n\log n)$ time. 

    \textbf{4:} Simply compute all points from the tree storing $A$ that are contained in some $[a_i,b_i]$ and output a tree storing them.

    \textbf{4:} Simply remove all points from the tree storing $A$ that are contained in some $[a_i,b_i]$.
\end{proof}

\begin{lemma}\label{lem:numberOfAtomic}
    Let $e$ be an edge of $S$. Let $A$ be a set of $\O(n^3)$ points in $[0,1]$ that is stored in a tree. There are $\O(n^2)$ molecular intervals $\molecular{}$ of $e$ and $P$, and they as well as their initial weights $\molW{A}(m)=\molW{A}(m)$ can be computed in $\O(n^2\log n)$ time.
\end{lemma}
\begin{proof}
    Straight forward, as for any interval $[a,b]$ the value $|A\cap[a,b]|$ can be computed in $\O(\log n)$ time.
\end{proof}

\begin{algorithm}[htp]
\caption{Covering a finite subset of $[0,1]$}\label{alg:covera}
\begin{algorithmic}[1]
\Procedure{CoverA}{$A_S$, $A\subset[0,1]$, $\{\edgeseqs{e}\mid e\text{ of }S\}$,$\{W_{e}\mid e\text{ of }S, \molW{A}(m)\neq 0\}$}
    \State Compute all Type $(I)$-curves and their proxy coverage $\proxyCov_{A_S}(\cdot)$
    \State For every $e$ compute $\molW{A}(m)$ and the midpoint of $m$ for every $m\in W_e$
    \State $R\gets\emptyset$
    \For{$e$ of $S$}
        \For{$\sweepseq\in\edgeseqs{e}$}
            \State Compute $\candW{A}(e[s,t])=|A\cap \proxyCov_{A_S}(e[s,t])|$ via $\molW{A}((s,t))$ for $(s,t)\in W_e$
        \EndFor
    \EndFor
    \While{$A\neq\emptyset$}
        \State For every Type $(I)$-curve $\pi$ compute $\candW{A}(\pi)=|A\cap \proxyCov_{A_S}(\pi)|$
        \State Add $r\in\ConcreteCandidates{S}(P)$ maximizing $\candW{A}(r)=|A\cap \proxyCov_{A_S}(r)|$ to $R$
        \State $A'\gets A \cap \proxyCov_{A_S}(r)$
        \For{$e$ of $S$}
            \State For every $m\in W_e$ compute $\molW{A'}(m)$ and $W'_e=\{m\in W_e\mid \molW{A'}(m)\neq 0\}$
            \State Remove $m$ from $W_e$ if $\molW{A\setminus A'}(m)=\molW{A}(m)-\molW{A'}(m)=0$
            \For{$\sweepseq\in\edgeseqs{e}$}
                \State Compute $\candW{A'}(e[s,t])$ via $\molW{A'}((s,t))$ for every $(s,t)\in\sweepseq$
                \State Update $\candW{A}(e[s,t],\emptyset)\gets \candW{A}(e[s,t],\emptyset) - \candW{A'}(e[s,t],\emptyset)$ for every $(s,t)\in\sweepseq$
            \EndFor
        \EndFor
        \State $A\gets A\setminus A'$
    \EndWhile
    \State \Return $R$
\EndProcedure
\end{algorithmic}
\end{algorithm}

\subsection{Putting everything together}\label{sec:cubic}

\begin{lemma}\label{lem:finalCosts}
    Let $P$ be a polygonal curve of complexity $n$ and let $\Delta>0$ and $\ell\leq n$ be given. Let $S$ be a simplification of $P$. Let $A_S$ be the $4\Delta$-free space of $S$ and $P$. Let $A\subset[0,1]$ be a set of $\O(n^3)$ points. For every edge $e$ of $S$ let $\edgeseqs{e}$ be $\O(\log n)$ sweep-sequences, each of length $\O(n)$, together containing all $\O(n\log n)$ Type $(II)$- and $(III)$-subedges of $e$, and let $W_e$ be the set of molecular intervals $m$ such that $\molW{A}(m)\neq 0$. Let $\kopt$ be the size of the smallest set $\mathcal{C}^*\subset\curvespace{\ell}$ such that $\bigcup_{c\in\mathcal{C}^*}\Cov_P(c,\Delta)=[0,1]$. The algorithm \textsc{CoverA} from \Cref{alg:covera} computes a set $\mathcal{C}\subset\curvespace{\ell}$ of size $(48\ln(n)+64)\kopt$ such that $A\subset\proxyCov_{A_S}(C)$. Further, it does so in time
    \[\O(n^2\ell\log n^2 + |A|\log n + \kopt n^2\log^3n + \log n\sum_e|W_e|).\]

\end{lemma}
\begin{proof}

    Observe that at the beginning of the \textbf{while}-Loop $A$ coincides with the set $A\setminus\proxyCov_{A_S}(R)$. Hence by \Cref{thm:reduction} it suffices to show that during the \textbf{while}-Loop the element $r\in\ConcreteCandidates{S}(P)$ maximizing $\candW{A}(r)$ is correctly identified.

    By \Cref{thm:type1}, computing all Type $(I)$-curves and their proxy coverage takes a total of $\O(n^2\ell\log n\log \ell)$ time. As $A$ is stored in a tree, computing $\molW{A}(m)$ and the midpoint of $m$ for every $m\in W_e$ takes $\O(|W_e|\log n)$ time. By \Cref{thm:pointQuery} computing and storing $\candW{A}(e[s,t])$ for every $(s,t)\in\sweepseq$ for every $\sweepseq\in\edgeseqs{e}$ for an edge $e$ of $S$ takes a total of $\O(|W_e|\log n+n^2\log^2 n)$ time. Thus the precomputation steps before the \textbf{while}-Loop take a total of $\O(n^2\ell\log n\log \ell + \sum_e(|W_e|\log n+n\log n))$ time.

    Now let for the $i$th iteration of the \textbf{while}-Loop $R_i$ denote the set stored in $R$ and let $A_i$ denote the set stored in $A$ coinciding with $A\setminus\proxyCov_{A_S}(R_i)$. Let further $W_{e,i}$ denote the set of molecular intervals stored in $W_e$, and similarly $W'_{e,i}$ be the set of molecular intervals stored in $W'_e$.
    Assume for now that $\candW{A_i}(e[s,t])$ is stored correctly for every Type $(II)$- and $(III)$-subcurve. Then $r$ is identified correctly. Observe that for $i=0$ this is indeed the case.

    Observe that $W_{e_i}=\{m\in\molecular{}(e,P)\mid \molW{A_i}(m)\neq 0\}$ and $W'_{e_i}=\{m\in\molecular{}(e,P)\mid \molW{A_i\cap\proxyCov_{A_S}(r)}(m)\neq 0\}$. By \Cref{lem:theDestroyer} $W'_{e,i}$ and its weights can be computed in time $\O(|W_{e_i}\setminus W_{e,{i+1}}| + n\log n)$. Similarly $W_{e,i}$ is afterwards updated to store $W_{e,i+1}=\{m\in\molecular{}(e,P)\mid \molW{A_i\setminus \proxyCov_{A_S}(r)}(m)=\molW{A_{i+1}}(m)\neq 0\}$. By \Cref{lem:theDestroyer} this update takes $\O(|W_{e,i}\setminus W_{e,{i+1}}| + n\log n)$. Lastly by \Cref{thm:pointQuery} computing $\candW{A_i\setminus A_{i+1}}(e[s,t])$ for every $(s,t)\in\sweepseq$ for an $\sweepseq\in \edgeseqs{e}$ for an edge $e$ takes time $\O(|A_{i}\setminus A_{i+1}|\log n + n\log n|)$. Afterwards $\candW{A_{i+1}}(e[s,t])$ is stored for all Type $(II)$- and $(III)$-subcurves as by definition $\candW{A_{i+1}}(\cdot)=\candW{{A_i}}(\cdot)-\candW{{A_i \setminus A_{i+1}}}(\cdot)$ and ${A_i \setminus A_{i+1}}=A_i\cap\proxyCov_{A_S}(r)$. Updating $A_i$ takes total time $\O(|A_i\setminus A_{i+1}| + n\log n)$. 
    
    Throughout the entirety of $k$ rounds the \textbf{while}-Loop without the computation of $\candW{A}$ of Type $(I)$-subcurves takes 
    \begin{align*}
        &\O\left(\sum_{i=0}^{k}\left(|A_i\setminus A_{i+1}|\log n +n\log n +\left( \sum_e\sum_{\sweepseq\in\edgeseqs{e}}\left( |W_{e,i}\setminus W_{e,i+1}| + n\log n \right) \right)\right)\right)\\
        =\;\;&\O\left(|A|\log n +kn\log n +\sum_{i=0}^{k}\left(\log n \sum_e\left( |W_{e,i}\setminus W_{e,i+1}| + n\log n \right) \right)\right)\\
        =\;\;&\O\left(|A|\log n +kn\log n + kn^2\log^2 n+\log n \sum_e\sum_{i=0}^{k}\left( |W_{e,i}\setminus W_{e,i+1}|\right) \right)\\
        =\;\;&\O\left(|A|\log n + kn^2\log^2 n+\log n \sum_e|W_e| \right).
    \end{align*}
    Thus by \Cref{thm:reduction} the algorithm terminates after $\O(n^2\ell\log n^2 + |A|\log n + \kopt n^2\log^3n + \log n\sum_e|W_e|)$ time, concluding the proof.
\end{proof}

\begin{figure}
    \centering
    \includegraphics[width=0.5\linewidth]{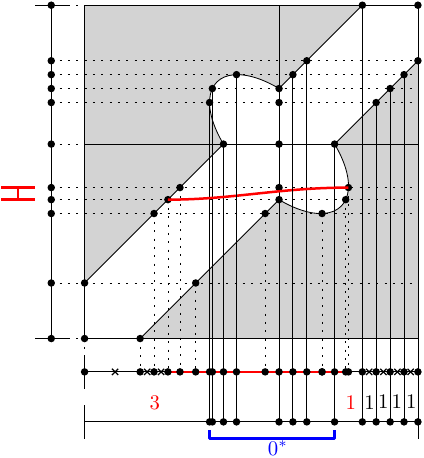}
    \caption{Illustration of the weights after a curve has been added to the partial solution $R$. Only constantly many molecular intervals per cell are updated but not set to $0$ (in red). All intervals that are set to $0$ are updated for the final time and can be removed from any search tree. These are further marked with a $^*$.}
    \label{fig:molecular_update}
\end{figure}

\thmMainClustering*
\begin{proof}

    The algorithm first computes a simplification $S$ of $P$. Next it computes the $4\Delta$-free space of $S$ and $P$ and with it $\extremal{}(A_S)$ as well as $\atomic{}(S,P)$ in $\O(n^3\log n)$ time. Next in total time $\O(n^2\log^2 n)$ time per edge $e$ it computes all molecular intervals $\molecular{}(e,P)$ and initializes their weights via \Cref{lem:numberOfAtomic}. Lastly it computes the data-structure from \Cref{thm:type1}.

    Lastly we compute $A$ of size $\O(n^3)$ via \Cref{obs:atomicnumber} and store it in a tree. We further compute the sets $W_e=\{m\in\molecular{}(S,P)\mid w_A(m)\neq 0\}$ for every edge $e$ of $S$, with $\sum_e|W_e|=\O(n^3)$.
    Then \Cref{lem:finalCosts} concludes the proof.
\end{proof}

\section{Improved Subtrajectory Covering}\label{sec:subcubic}

In this section we want to reduce the dependency on $n$. Observe that there are two steps in the algorithm above that require $\tildeO(n^3)$ time that are somewhat inevitable. The first is the computation of the discretization, that is $\atomic{}(S,P)$ may be of cubic size, and similarly the $n$ coarsenings via molecular intervals $\molecular{}(e,P)$ in total take $\O(n^3)$ time to compute. The later updates depend on the number of molecular intervals with non-zero weight. In this section we introduce a subproblem which is closely related to rank-selection \cite{frederickson1984matrix} in matrices in order to pick $\O(n^{\coarsity})$ interval boundaries of $\atomic{}(e,P)$ which are `well-separated' in time $\tildeO(n^{1+\coarsity})$, that is there are no two such boundaries defining a coarse interval that contain more than $\O(n^{3-\coarsity})$ intervals of $\atomic{}(e,P)$. The algorithm will then proceed and solve the set-cover instance defined by these subquadratically many points. The coverage of such a solution can be described by the union of $\O(\kopt n\log n)$ intervals, and thus at most $\O(\kopt n\log n)$ of these coarse intervals are not covered. As no interval contains too many intervals of $\atomic{}(e,P)$ It turns out that at most $\O(\kopt n^{4-\coarsity})$ molecular intervals have weight not $0$. These we determine in an output-sensitive manner, and then proceed with the algorithm described in \Cref{sec:cubic}. This results in an algorithm with running time roughly $\tildeO(\kopt n^{5/2})$ by setting $\coarsity=3/2$. By multiplicatively testing for $\kopt $ and being slightly more careful how $\alpha$ is chosen the running time improves to roughly $\tildeO(\kopt ^{\frac{1}{2}}n^{\frac{5}{2}})$ which is truly subcubic if $\kopt\in\O(n^{1-\eps})$.

\subsection{Subquadratic coarsening of atomic intervals}

\begin{theorem}\label{thm:coarsesubset}
    Let $n$ lists $L_i$ be given, each containing $m_i$ sorted values such that all values are distinct and in every list $L_i$ and every $j$ identifying the item at position $j$ takes $\O(\log m_i)$ time and also for given $v$ determining the maximal index $j$ such that the $j$th item is less than $v$ takes $\O(\log m_i)$ time. Then for every $K\leq\sum_im_i$ in $\O(Kn\log \max_im_i)$ time one can determine $\O(K)$ values $v_1,\ldots$ such that
    \begin{compactenum}
        \item $v_i<v_{i+1}$ for all $i$,
        \item for every $x\in\bigcup_iL_i$ there is an $i$ such that $x\in[v_i,v_{i+1}]$ and 
        \item $|[v_i,v_{i+1}]\cap\bigcup_{i}L_i|=O\left(\frac{\sum_im_i}{K}\right)$ for all $i$.
    \end{compactenum}
\end{theorem}
\begin{proof}
    Initially identify the minimal $v$ and maximal $\hat{v}$ value in $\bigcup_{i}L_i$ in $\O(n\log \max_im_i)$ time. We define a subproblem-instance via $v$ and $\hat{v}$ by $n$ intervals of indices describing the interval of $L_i$ that lies between $v$ and $\hat{v}$. Initially all intervals have the form $[1,m_i]$. The goal now is to split this subproblem into smaller subproblems until the desired distance is reached, that is the sum of lengths of intervals in each subproblem that is not split further is at most $O\left(\frac{\sum_im_i}{K}\right)$.
    
    Let now a subproblem-instance be given, defined by values $l$, $r$, and $n$ intervals with length $l_i$.
    If $\sum_il_i=O\left(\frac{\sum_im_i}{K}\right)$, then there is no need to split it.
    
    If instead $\sum_il_i\leq 20n$, then we simply sort the $n$ intervals in total time $\O(n\log n)$ time and identify the value $m$ splitting the sorted list into half. From this construct two subproblems in $\O(n\log n)$ time where the sum of intervals lengths is halved.
    
    If $\sum_il_i> 20n$, then for every list $L_i$ identify $5$ values splitting $L_i$ into $5$ pieces all containing either $\lfloor  m_i/5\rfloor$ or $\lceil  m_i/5\rceil$ values, this results in $5$ intervals, which are endowed with a weight $w_I$ corresponding to the number of values from $L_i$ that lie in it. This results in a total set $\mathcal{I}$ of $\O(n)$ which can be computed in $\O(n\log \max_im_i)$ time, and so can its arrangement. For every interval in the arrangement pick its mid point as the potential splitting value, and collect them in a sorted list $C$. For every $c\in C$ define
    \[w(c)=\sum_{I\in\mathcal{I}, I<c}w_I - \sum_{I\in\mathcal{I}, I>c}w_I.\]
    All these values can be computed in $\O(n\log \max_i m_i)$ time as $w(c)$ and $w(c')$ differ by exactly two $w_I$ if $c$ and $c'$ lie in neighboring cells of the arrangement, thus we can sweep through the arrangement. Now identify the last $c\in C$, such that $w(c)\geq 0$, that is in particular $w(c)\leq \lceil \max_il_i/5\rceil\leq \lceil\sum_il_i/5\rceil\leq\sum_il_i/4$. And hence
    \[\sum_{I\in\mathcal{I}, I>c}w_I\leq\sum_{I\in\mathcal{I}, I<c}w_I\]
    and
    \[\sum_{I\in\mathcal{I}, I<c}w_I\leq\sum_{I\in\mathcal{I}, I>c}w_I + \sum_il_i/4\]
    Further observe that
    \begin{align*}
        \sum_{I\in\mathcal{I}, I<c}w_I + \sum_{I\in\mathcal{I}, I>c}w_I&\geq\sum_{I\in\mathcal{I}}w_I - \sum_{I\in\mathcal{I}, c\in I}w_I\geq\sum_il_i - \sum_i\lceil l_i/5\rceil\\
        &\geq\sum_il_i - \sum_i l_i/5 - n\geq \sum_il_i - \sum_i l_i/4 = 3/4\sum_il_i
    \end{align*}
    Thus we conclude that
    \[\sum_{I\in\mathcal{I}, I>c}w_I\geq1/4 \sum_il_i\]
    and
    \[\sum_{I\in\mathcal{I}, I<c}w_I\geq1/4  \sum_il_i.\]
    Hence splitting at $c$ will result in two subproblems, each of which containing at most $3/4\sum_il_i$ points. These subproblems can similarly be constructed in $\O(n\log \max m_i)$ time. As each split decreases the sum of interval lengths by a constant fraction, after $\O(\log K)$ levels of splits each subproblem will have at most $O\left(\frac{\sum_im_i}{K}\right)$ points concluding the proof.
\end{proof}

\begin{lemma}\label{lem:coarseIntervals}
    For every $\coarsity\in[0,3]$ in $\O(n^{1+\coarsity}\log n)$ time one can determine $\O(n^{\coarsity})$ intervals partitioning $\atomic{}(S,P)$, each containing at most $\O(n^{3-\coarsity})$ intervals of $\atomic{}(S,P)$.
\end{lemma}

We will call a set of intervals that partitioning $\atomic{}(S,P)$ such that any interval contains $\O(n^{3-\coarsity})$ intervals of $\atomic{}(S,P)$ a set of \emph{$\coarsity$-coarse} intervals.

\begin{proof}
    For a fixed edge $e$ we will implicitly provide two lists $L_u$ and $L_l$ such that 
    \begin{compactenum}
        \item $L_u\cup L_l = \molecular{}(e,S)$,
        \item $|L_u|$ and $|L_l|$ is known
        \item for either list the value at position $j$ can be computed in $\O(\log n)$ time\label{item:3}
        \item for any value $v$ the maximal index $j$ such that the value at position $j$ is less than $v$ can be computed in $\O(\log n)$ time.\label{item:4}
    \end{compactenum}
    To construct $L_u$ we compute in every cell the at most $4$ closures of the at most $4$ pieces of the free space boundary (i.e. an ellipse) that lie in the interior of the cell and are $x$- and $y$-monotone. From this we pick the at most $2$ pieces that lie above (in $y$-direction) the free space. Observe, that for these $\O(n)$ pieces from every cell it holds that their projections onto the $x$-axis are interior-disjoint. For every piece we compute the interval of $\extremal{}(e,P)$ in $\O(\log n)$ time that lie in the projection onto the $y$-axis of each piece. These form the list $L_u$. Observe that we can compute $|L_u|$ in $\O(n)$ time and store it. Further we can endow these $\O(n)$ pieces such that \Cref{item:3} and \Cref{item:4} also hold. We similarly compute $L_l$ via the lower, instead of the upper pieces of the boundary. Observe, that the union over both lists form $\molecular{}(e,S)$. Further, the only values that occur with multiplicity in each list correspond to intersections with the free space boundary, and can be made unique in $\O(n)$ time.

    Computing this for every edge in $\O(n^2\log n)$ time leaves us with $\O(n)$ lists each containing at most $\O(n^2)$. By symbolic perturbation via the list-index we define a total-order on all $\O(n^3)$ values which allows applying \Cref{thm:coarsesubset} for $n^\coarsity$ resulting in $\O(n^\coarsity)$ interval boundaries of $\atomic{}(S,P)$ in $\O(n^{1+\coarsity}\log n)$ time such that between two such values at most $\O(n^{3-\coarsity})$ intervals of $\atomic{}(S,P)$ lie in between them.
\end{proof}

Setting $\coarsity=3/2$ balances the number of $\coarsity$-coarse intervals with the number of atomic intervals in each $\coarsity$-coarse interval. In fact in the end we will set $\coarsity\approx 3/2$, however, we can improve the running time by increasing $\coarsity$ ever so slightly depending on $\kopt $. As we do not know $\kopt $, we test for it by doubling a guess $K$ of $\kopt $ every time the algorithm does not produce a solution of size $K\log n$. Thus the algorithm in the end will compute different coarsenings for values $\coarsity\in[3/2,2]$, hence we analyze everything depending on the variable $\coarsity$.

\subsection{Subtrajectory Covering without explicit atomic intervals}\label{sec:combinatorialdescription}

Thus far, we have used implicit atomic/molecular intervals to construct an explicit small subset of atomic intervals. We now apply this to get an initial coarse solution.

\begin{lemma}\label{lem:uncoverableCompute}
    For every $\coarsity\in[0,3]$ and every $K$ let $C\subset\ConcreteCandidates{S}(P)$ be a solution covering all midpoints of $\coarsity$-coarse intervals of size $\O(K \log n)$. Then there are at most $\O(K n^{4-\coarsity}\log n)$ atomic intervals and  $\O(K n^{4-\coarsity}\log n + K n^2\log n)$ molecular intervals that are not contained in $\proxyCov_{A_S}(C)$. Further they can be computed in $\O(K n^{4-\coarsity}\log^2 n +K n^2\log^2n)$ time.
\end{lemma}
\begin{proof}
    First observe, that there are at most $\O(K n^{4-\coarsity}\log n)$ atomic intervals that are not in $\proxyCov_{A_S}(C)$, as $\proxyCov_{A_S}(C)$ does not contain at most $\O(K n\log n)$ $\alpha$-coarse intervals and hence at most $\O(K n^{4-\alpha}\log n)$ atomic intervals. Similarly, the set of molecular intervals that are not contained in $\proxyCov_{A_S}(C)$ consists of all molecular intervals that have a boundary that is not in $\proxyCov_{A_S}(C)$, or that contain some interval boundary of $\proxyCov_{A_S}(C)$. Thus in total there are $\O(K n^{4-\alpha}\log n + K n^2\log n)$ such molecular intervals.

    To compute them observe that $\proxyCov_{A_S}(C)$ consists of $\O(K n\log n)$ disjoint intervals, and thus so does $[0,1]\setminus \proxyCov_{A_S}(C)$. Let $\mathcal{I}$ be the set of $\O(K n \log n)$ intervals that result from the disjoint intervals in $[0,1]\setminus \proxyCov_{A_S}(C)$ after intersecting them with the preimage of every edge. The intervals in $\mathcal{I}$ are not necessarily disjoint, but the intersection of any two intervals is in at most one point. Let $[a,b]\in\mathcal{I}$ lie on edge $i$ of $P$. We compute the sets represented as $\O(1)$ contiguous intervals in $\extremal{}(A_e)$
    \begin{compactenum}
        \item $S_{e,1}=\{l_i(s)\mid s\in\extremal{}(A_e) : l_i(s)\in[a,b]\}$,
        \item $S_{e,2}=\{r_i(s)\mid s\in\extremal{}(A_e) : r_i(s)\in[a,b]\}$,
        \item $S_{e,3}=\{l_i(s)\mid s\in\extremal{}(A_e) : l_i(s) \leq a\}$ and $S_{e,4}=\{r_i(s)\mid s\in\extremal{}(A_e) : r_i(s) \leq a\}$, and
        \item $S_{e,5}=\{l_i(s)\mid s\in\extremal{}(A_e) : l_i(s) \geq b\}$ and $S_{e,6}=\{r_i(s)\mid s\in\extremal{}(A_e) : r_i(s) \geq b\}$.
    \end{compactenum}
    in total time $\O(\log n)$ via binary searches over $\extremal{}(A_e)$. Then the union $\{a,b\}\cup\bigcup_{e\in S}S_{e,1}\cup S_{e_2}$ is precisely the set of all interval boundaries of atomic intervals that lie in $[a,b]$. These can thus be computed in $\O((|[a,b]\cap\atomic{}(S,P)|+1)\log n)$ time. Thus overall computing all atomic intervals that intersect $[0,1]\setminus \proxyCov_{A_S}(C)$ can be done in $\O(K n^{4-\alpha}\log^2 n +Kn\log^2n)$ time.
    
    Similarly $\{\max(S_{e,3}\cup S_{e,4}),\max(S_{e,5}\cup S_{e,6})\}\cup S_{e,1}\cup S_{e,2}$ contain all interval boundaries of molecular intervals of $e$ and $P$ that intersect $[a,b]$ and can be computed in time $\O((|\{m\in\molecular{}(e,P)\mid m\subset [a,b]\}| + 1)\log n)$ time. Thus computing all molecular intervals that intersect $[0,1]\setminus \proxyCov_{A_S}(C)$ can be done in $\O(K n^{4-\alpha}\log^2 n + K n^2\log^2 n)$ time.
\end{proof}

\begin{algorithm}
\caption{Covering $[0,1]$ in subcubic time}\label{alg:covera2}
\begin{algorithmic}[1]
\Procedure{CoverAFast}{$P$,$\Delta$,$\ell$}
    \State Compute a simplification $S$ of $P$
    \State Compute the $4\Delta$-free space $A_S$ of $S$ and $P$
    \State Compute the extremal $y$-coordinates $\extremal{}(A_e)$ for every edge $e$
    \State Compute the sweep-sequences $\mathcal{S}_e$ for every edge $e$
    \State Compute the proxy coverage for every Type $(I)$-curve and provide them to \textsc{CoverA}
    \State $K\gets 1$, $\mathrm{covers}\gets\textsc{False}$, $\lambda\gets(48\ln|P| + 64)$
    \While{$\mathrm{covers}$ is \textsc{False}}
        \State $\alpha\gets \frac{3}{2} + \frac{\log K}{2\log n} + \frac{\log\log n}{\log n}$, $\mathrm{covers}\gets\textsc{True}$ and $R\gets\emptyset$
        \State Compute a set of $\alpha$-coarse intervals and from them their midpoints $A$
        \State Compute set of molecular intervals $W_e$ with $w_A(\cdot)\neq 0$ for every edge $e$ of $S$
        \If{\textsc{CoverA}($A_S$,$A$,$\{\edgeseqs{e}\}$,$\{W_e\}$) does not terminate after $\lambda K$ rounds}
            \State $K\gets 2K$, $\mathrm{covers}\gets\textsc{False}$
        \Else
            \State $R\gets\text{\textsc{CoverA}($A_S$,$A$,$\{\edgeseqs{e}\}$,$\{W_e\}$)}$
            \State Compute all atomic intervals in $\atomic{}(S,P)\setminus\proxyCov_{A_S}(R)$ and their midpoints $A$
            \State Compute set of molecular intervals $W_e$ with $w_A(\cdot)\neq 0$ for every edge $e$ of $S$
            \If{\textsc{CoverA}($A_S$,$A$,$\{\edgeseqs{e}\}$,$\{W_e\}$) does not terminate after $\lambda K$ rounds}
                \State $K\gets 2K$, $\mathrm{covers}\gets\textsc{False}$
            \Else
                \State $R\gets R \cup \text{\textsc{CoverA}($A_S$,$A$,$\{\edgeseqs{e}\}$,$\{W_e\}$)}$
            \EndIf
        \EndIf
        
    \EndWhile
    \State \Return $R$
\EndProcedure
\end{algorithmic}
\end{algorithm}

\coverafast*

\begin{proof}
    Let $\kopt$ be the size of the smallest set $\mathcal{C}^*\subset\curvespace{\ell}$ such that $\bigcup_{c\in\mathcal{C}^*}\Cov_P(c,\Delta)=[0,1]$.
    
    Line $1$ to $6$ take time $\O(n^2\ell\log n\log\ell)$ by \Cref{thm:thijs-simplification}, the fact that the extremal points in all cells of the free space can be computed in $\O(n^2)$ time, \Cref{lem:allsweep} and \Cref{thm:type1}.

    Now let $K$ be as in the first Line of the \textbf{while}-Loop in Line $8$. We show that within one iteration of the \textbf{while}-Loop the algorithm compute a solution of size $K\log n$ or correctly determines that $K<k_\delta$ and sets $K\gets 2K$. Observe that initially $K=1\leq k_\Delta$. 
    Let $\coarsity = \frac{3}{2} + \frac{\log K}{2\log n} + \frac{\log\log n}{\log n}$ and let $\lambda=(48\ln n + 64)\geq 16(\ln(16n^3)+1)$.
    By \Cref{lem:coarseIntervals}, the algorithm first computes a set of $\alpha$-coarse intervals correctly and then attempts to cover them in $\lambda K$ rounds of \textsc{CoverA} from \Cref{alg:main}. If it does not terminate after $\lambda K$ rounds, then $\lambda K < \lambda \kopt$ and in particular $K<\kopt$, in which case we correctly set $K\gets 2k$ and restart the \textbf{while}-Loop. Otherwise \textsc{CoverA} returns a solution of size $\lambda \kopt$ covering all midpoints of the $\alpha$-coarse intervals. 

    Next, via \Cref{lem:uncoverableCompute}, the algorithm determines all uncovered atomic intervals and the molecular intervals which contain these points. Next the algorithm invokes \textsc{CoverA} again, this time with the set of midpoints of uncovered atomic intervals, and the set of molecular intervals that contain at least one of these midpoints. If it does not terminate after $\lambda K$ rounds, then $\lambda K < \lambda\kopt$ and in particular $K<\kopt$, in which case we correctly set $K\gets 2k$ and restart the \textbf{while}-Loop. Otherwise \textsc{CoverA} returns a solution of size $\lambda \kopt$ covering all midpoints of this subset of atomic intervals.

    Thus within the \textbf{while}-Loop the algorithm correctly determines that either $K<\kopt$ and restarts with $K\gets2K$ or outputs a solution of size $(96\ln(n)+128)\kopt$, thus the algorithm correctly computes a solution of claimed size.

    As $K\in\O(\kopt)$ and thus $K\in\O(n)$, the running time of the \textbf{while}-Loop is 
    \begin{align*}
        &\O(n^{1+\coarsity}\log n + Kn^{4-\coarsity}\log^3 n+Kn^2\log^2 n)\\
        &=\O(n^{1+\coarsity}\log n + n^{4-\coarsity+\frac{\log K}{\log n}}\log^3n+Kn^2\log^2 n)\\
        &=\O(K^{\frac{1}{2}}n^{\frac{5}{2}}\log^2 n+K^{\frac{1}{2}}n^{\frac{5}{2}}\log^2n+Kn^2\log^2n)=\O(K^{\frac{1}{2}}n^{\frac{5}{2}}\log^2n).
    \end{align*}
    And thus overall the running time of the algorithm is 
    \begin{align*}
        &\O\left(n^2\ell\log^2n+\sum_{K=1}^{\log(\kopt )}\left(\left(2^K\right)^\frac{1}{2}n^{\frac{5}{2}}\log^2n\right)\right)\\
        &=\O\left(n^2\ell\log^2n+\kopt ^{\frac{1}{2}}n^{\frac{5}{2}}\log^2n\right)\\
        &=\O\left(n^2\ell\log^2n+\kopt ^{\frac{1}{2}}n^{\frac{5}{2}}\log^2n\right).\qedhere
    \end{align*}
\end{proof}

\begin{proposition}
    If $\kopt =\O(n^{1-\eps})$ and $\ell=\O(n^{1-\eps})$ then the running time is subcubic namely $\tildeO(n^{3-\eps/2})$.
\end{proposition}

\section{Subtrajectory Coverage Maximization}\label{sec:maximization}

The goal of this section is the following theorem.

\thmMainCoverage*

Throughout this section we assume that we are given $P$ and we have computed the simplification $S$ via \Cref{thm:thijs-simplification} and an $(1+\eps,\Delta)$-free space $A_S$ via \Cref{thm:approxFreeSpace} consisting of convex polygons of complexity $O(\eps^{-2})$ in each cell. 

The quality of the algorithm hinges upon the following theorem.

\begin{theorem}\label{thm:coverageReduction}
    Let $P$ be curve of complexity $n$, let $\Delta$ and $\ell$ be given. Let $S$ be a simplification of $P$. For ever edge $e$ of $P$ let $A_e$ be an $(1+\eps)$-approximate piece-wise linear $\Delta$-free space for $e$ and $P$. Let $\ConcreteCandidates{S}(P)$ be the set of all $(I)$-, $(II)$- and $(III)$-subcurves of $S$ induced by the set of all $y$-coordinates of extremal points and polygon-vertices of all $A_e$. Any algorithm that iteratively adds the curve $c$ among $\ConcreteCandidates{S}(P)$ to $R$ maximizing
    \[\left\|\proxyCov_{A_S}(c)\setminus\left(\bigcup_{r\in R}\proxyCov_{A_S}(r)\right)\right\|,\]
    computes after $k$ rounds a set $R\subset\bX_\ell^d$ such that for any other set $C^*$ of cardinality $k$ it holds that \[\|\Cov_P(R,(4+\eps)\Delta)\|\geq\frac{e-1}{16e}\|\Cov_P(C^*,\Delta)\|\]
\end{theorem}
\begin{proof}
    Let $\textsc{Opt}^*_{\Cov,k}$ be the subset of $\mathcal{S}$ of size $k$ maximizing $\|\Cov_{A_S}(\cdot)\|$. Let $\textsc{Opt}_{\Cov,k}$ be some subset of $\bX^d_\ell$ of size $k$ maximizing $\|\Cov_P(\cdot,\Delta)\|$. Then by \Cref{lem:supersetfreespace} we know that
    \[\|\Cov_{A_S}(\textsc{Opt}^*_{\Cov,k})\|\geq 1/8 \|\Cov(\textsc{Opt}_{\Cov,k},\Delta)\|.\]
    Let $\overleftarrow{\textsc{Opt}^*_{\Cov,k}}$ be the set of reversed sub-edges in $\textsc{Opt}^*_{\Cov,k}$. Then by \Cref{lem:proxyapprox}
    \[\|\proxyCov_{A_S}(\textsc{Opt}^*_{\Cov,k}\cup\overleftarrow{\textsc{Opt}^*_{\Cov,k}})\|\geq \|\Cov_{A_S}(\textsc{Opt}^*_{\Cov,k})\|.\]

    Let further $\textsc{Opt}^*_{\proxyCov,k}$ and $\textsc{Opt}^*_{\proxyCov,2k}$ be the sets of size $k$ and $2k$ in $\ConcreteCandidates{S}(P)$ respectively, maximizing $\|\proxyCov_{A_S}(\cdot)\|$. Then by sub-additivity 
    \[2\|\proxyCov_{A_S}(\textsc{Opt}^*_{\proxyCov,k})\|\geq\|\proxyCov_{A_S}(\textsc{Opt}^*_{\proxyCov,2k})\|\geq \|\proxyCov_{A_S}(\textsc{Opt}^*_{\Cov,k}\cup\overleftarrow{\textsc{Opt}^*_{\Cov,k}})\|. \]

    Lastly as $\|\proxyCov_{A_S}(\cdot)\|$ is submodular, by standard greedy submodular function maximization arguments \cite{Nemhauser1978,krause11} it holds that
    \[\|\proxyCov_{A_S}(R)\|\geq \frac{e-1}{e}\|\proxyCov_{A_S}(\textsc{Opt}^*_{\proxyCov,k})\|.\]
    And finally 
    \[\|\Cov_P(R,(4+\eps)\Delta)\| \geq\|\Cov_{A_S}(R)\| \geq \|\proxyCov_{A_S}(R)\|.\]
    Thus overall, we observe that after $k$ iterations it holds that 
    \[\|\Cov_P(R,(4+\eps)\Delta)\|\geq\frac{e-1}{16e}\|\Cov_P(\textsc{Opt}_{\Cov,k},\Delta)\|.\qedhere\]
\end{proof}

\subsection{\boldmath Maintaining a partial solution and Type $(I)$-subcurves}

In the following sections we describe, how to maintain a partial solution $R$, and how to compute the element in $\ConcreteCandidates{S}(P)$ which maximizes \[\left\|\proxyCov_{A_S}(c)\setminus\left(\bigcup_{r\in R}\proxyCov_{A_S}(r)\right)\right\|.\]

A partial solution $R$ consists of at most $k$ elements from $\ConcreteCandidates{S}(P)$. At any point in the algorithm we maintain a set $\mathcal{I}$ of $\O(kn)$ disjoint intervals, the union of which describe $\proxyCov_{A_S}(R)$.

\begin{lemma}\label{lem:solution_datastructure2}
    Let $I$ be a set of $n$ pairwise disjoint intervals. There is a data-structure that allows computing

\[\left\|[l,r]\setminus\left(\bigcup_{i\in I_R}i\right)\right\|\]
    
    in $\O(\log n)$ time for arbitrary $[l,r]$. Further we can either remove an interval from $I$, or if $[l,r]\cap \left(\bigcup_{i\in I}i\right)=\emptyset$, then $[l,r]$ can be added to $I$ correctly updating the data-structure $\O(\log n)$ time.
\end{lemma}
\begin{proof}
    Simply store the disjoint intervals $I_R$ in an interval-tree and for every interval store the length of this interval. This value can be computed in $\O(1)$ time.
\end{proof}

\begin{lemma}
    Given $\proxyCov_{A_S}(S[s,t])$ for every Type $(I)$-subcurve $S[s,t]$ of $S$, and an interval set $\mathcal{I}$ of $\O(kn)$ intervals representing $\proxyCov_{A_S}(R)$ stored in the data-structure from   \Cref{lem:solution_datastructure2} one can identify the Type $(I)$-subcurve of $S$ maximizing 
    \[\left\|\proxyCov_{A_S}(c)\setminus\left(\bigcup_{r\in R}\proxyCov_{A_S}(r)\right)\right\|\]
    in time $\O(n^2\log \ell\log n)$.
\end{lemma}
\begin{proof}
    Let $S[s,t]$ be a Type $(I)$-subcurve of $S$.
    For every interval $I$ in $\proxyCov_{A_S}(S[s,t])$ compute $\|I\setminus\bigcup_{j\in\mathcal{I}}j\|$ via the data-structure from \Cref{lem:solution_datastructure2} representing $\mathcal{I}$ in time $\O(\log kn) = \O(\log n)$. From this we compute $\|\proxyCov_{A_S}(S[s,t])\setminus\bigcup_{j\in\mathcal{I}}j\|$ in total time $\O(n\log n)$ time. As there are a total of $\O(n\log\ell)$ Type $(I)$-subcurves of $S$ the claim follows.
\end{proof}

\subsection{\boldmath Type $(II)$- and $(III)$-subcurves}

We say $\nonprob(e[s,t])=\redGlobG(e[s,t])\cup(\localG(e[s,t])\setminus\bad(e[s,t]))$. In \Cref{sec:symbolicRepresentation} we have seen that $\nonprob(e[s,t])$ can be maintained throughout the scan of a sweep-sequence. Further it holds that \[\proxyCov(e[s,t])=\bigsqcup_{(i,j)\in\nonprob(e[s,t])}[\hat{l}_{i,e[s,t]}(s),\hat{r}_{j,e[s,t]}(t)].\]

Assume that one is given $\nonprob(e[s,t])$. Let further $I_R$ be a set of intersection free intervals representing $\proxyCov$ of a partial solution $R$. Define \(\hat{l}_{i,e[s,t]}(s,I_R)=\|[\hat{l}_{i,e[s,t]}(s),i_1]\setminus\bigcup_{[a,b]\in I_R}[a,b]\|,\) where $i_1$ defines the $x$-coordinate of the $(i+1)$th vertex of $P$. Clearly $\hat{l}_{i,e[s,t]}(s,I_R)$ can be computed in $\O(\log n)$ time. Similarly define $\hat{r}_{j,e[s,t]}(s,I_R)=\|[j_0,\hat{r}_{j,e[s,t]}(s)]\setminus\bigcup_{[a,b]\in I_R}[a,b]\|$, where $j_0$ defines the $x$-coordinate of the $j$th vertex of $P$.
Then for any $(i,j)\in \nonprob(e[s,t])$ we see that \[\left\|[\hat{l}_{i,e[s,t]}(s),\hat{r}_{j,e[s,t]}(t)]\setminus\bigcup_{[a,b]\in I_R}[a,b]\right\| = \hat{l}_{i,e[s,t]}(s,I_R) + \hat{r}_{j,e[s,t]}(s,I_R) + L(i,j,I_R),\]
where $L(i,j)=\|[i_1,j_0]\setminus\bigcup_{[a,b]\in I_R}[a,b]\|$ if $i_1\leq j_0$, otherwise if $i=j$ then $L(i,j)=-\|[i_0,i_1]\setminus\bigcup_{[a,b]\in I_R}[a,b]\|$.

This in turn implies that 
\[\left\| \proxyCov(e[s,t])\setminus\bigcup_{[a,b]\in I_R}[a,b] \right\| = \sum_{(i,j)\in \nonprob(e[s,t])}\left(\hat{l}_{i,e[s,t]}(s,I_R) + \hat{r}_{j,e[s,t]}(t,I_R) + L(i+1,j,I_R)\right).\]

Observe, that $\sum_{(i,j)\in \nonprob(e[s,t])}L(i+1,j,I_R)$ is completely independent of the exact values $s$ and $t$, and only depends on $\nonprob(e[s,t])$ and $I_R$, and can be maintained during the sweep, as each value $L(i,j,I_R)$ can be computed in $\O(\log n)$ time, and the coverage description is maintained correctly with a total of $\O(n)$ updates to $\nonprob(e[s,t])$ via \Cref{thm:maintainSymbolic}.

We now show how to maintain linear functions representing the Lebesgue-measure of the proxy coverage along a sweep-sequence, given that the approximate free space is described by piece-wise linear functions. For this we first show how to maintain the functions \(l_i(s,I_R)=\|[l_i(s),i_1]\setminus\bigcup_{[a,b]\in I_R}[a,b]\|\) and $r_j(s,I_R)=\|[j_0,r_j(s)]\setminus\bigcup_{[a,b]\in I_R}[a,b]\|$.

\begin{lemma}\label{lem:primitiveiterativesum}
    Let $e$ be a given edge. 
    Let $I=\{(s_a,t_a),\ldots,(s_b,t_b)\}$ be a continguous interval in the sweep-sequence $\sweepseq=\{(s_1,t_1),\ldots,(s_m,t_m)\}\in\edgeseqs{e}$ such that for any $(s_k,t_k)\in\sweepseq$ it holds that $l_i(s_k)\neq\infty$. There are values $m_{(s_i,t_i),I}$ and $b_{(s_i,t_i),I}$ for every $(s_i,t_i)\in \sweepseq$ such that for $(s_r,t_r),(s_{r+1},t_{r+1})\in \sweepseq$ consecutive it holds that
    \[\mathbbm{1}_{(s_r,t_r)\in I}\cdot l_i(s_r) = \mathbbm{1}_{(s_{r+1},t_{r+1})\in I}\cdot l_i(s_{r+1}) + \left(\sum_{k\geq r}m_{(s_k,t_k),I}\right)(s_{r+1}-s_r)+\sum_{k\geq r}b_{(s_k,t_k),I}.\]
    Further there are only $\O(\eps^{-2})$ non-zero such values, and they can be computed in $\O(\eps^{-2}\log n\log\eps^{-1})$.
\end{lemma}
\begin{proof}
    In this case we set $b_{(s_b,t_b),I}=l_i(\eps_j)$ and $m_{(s_b,t_b),I}$ as the slope of the polygon at the point $l_i(s_b)$. We then traverse the polygon in a counter-clockwise manner until we hit the next vertex of the polygon at which the slope changes from $m=m_{(s_b,t_b),I}$ to $m'$. By definition of $\extremal{}(A_e)$ this vertex lies at some height $s_k\in \extremal{}(A_e)$ We then set $m_{(s_k,t_k),I}\gets m'-m$. We continue until we pass the height $(s_a,t_a)$. We set $m_{(s_a,t_a),I}$ such that $\sum_{g\geq i}m_{(s_k,t_k),I,I_R}=0$ and set $b_{\eps_i,I}=-l_i(\eps_i)$. Observe that $\sum_{k\geq i}m_{(s_k,t_k),I,I_R}$ is the local slope of the polygon and thus correctness follows. Further observe that these values can be identified in total time $\O(\log n)$ time as the only steps that do not take constant time are the look-ups of the indices of $\eps_k$.
\end{proof}
\begin{lemma}\label{lem:primitive2iterativesum}
    Let $e$ be an edge of $S$. Let $I=\{(s_a,t_a),\ldots,(s_b,t_b)\}$ be a continguous interval in the sweep-sequence $\sweepseq=\{(s_1,t_1),\ldots,(s_m,t_m)\}\in\edgeseqs{e}$ such that for any $(s_k,t_k)\in\sweepseq$ it holds that $l_i(s_k)\neq\infty$. Let $J\subset[0,1]$. There are values $m_{(s_i,t_i),I,J}$ and $b_{(s_i,t_i),I,J}$ for every $(s_i,t_i)\in \sweepseq$ such that for $(s_r,t_r),(s_{r+1},t_{r+1})\in \sweepseq$ consecutive it holds that
    \begin{align*}
        \mathbbm{1}_{(s_r,t_r)\in I}\cdot \|[l_i(s_r),i_1]\cap J\|&=\mathbbm{1}_{(s_{r+1},t_{r+1})\in I}\cdot \|[l_i(s_{r+1}),i_1]\cap J\|\\
        &+ \left(\sum_{k\geq r}m_{(s_k,t_k),I,J}\right)(s_{r+1}-s_r)+\sum_{k\geq r}b_{(s_k,t_k),I,J}.
    \end{align*}
    Further there are only $\O(\eps^{-2})$ non-zero such values, and they can be computed in $\O(\eps^{-2}\log n\log\eps^{-1})$.
\end{lemma}
\begin{proof}
    Observe that for $J=[j_0,j_1]$ the subcurve $P[j_0,j_1]$ has complexity $2$ and thus acts like an edge of $P$, and thus the free space $A_e$ restricted to $e$ and $P[j_0,j_1]$ may act like a cell of the free space. Inside this cell we define the function $l$ similar to $l_i$. Then the claim is an immediate consequence of \Cref{lem:primitiveiterativesum}.
\end{proof}

\begin{corollary}\label{lem:iterativesum}
    Let $e$ be a given edge. 
    Let $I=\{(s_a,t_a),\ldots,(s_b,t_b)\}$ be a continguous interval in the sweep-sequence $\sweepseq=\{(s_1,t_1),\ldots,(s_m,t_m)\}\in\edgeseqs{e}$ such that for any $(s_k,t_k)\in\sweepseq$ it holds that $l_i(s_k)\neq\infty$. Let $\mathcal{I}_R$ be a set of disjoint intervals representing a partial solution. There are values $m_{(s_i,t_i),I,I_R}$ and $b_{(s_i,t_i),I,I_R}$ for every $(s_i,t_i)\in \sweepseq$ such that for $(s_r,t_r),(s_{r+1},t_{r+1})\in \sweepseq$ consecutive it holds that
    \[\mathbbm{1}_{(s_r,t_r)\in I}\cdot l_i(s_r,I_R) = \mathbbm{1}_{(s_{r+1},t_{r+1})\in I}\cdot l_i(s_{r+1},I_R) + \left(\sum_{k\geq r}m_{(s_k,t_k),I,I_R}\right)(s_{r+1}-s_r)+\sum_{k\geq r}b_{(s_k,t_k),I,I_R}.\]
    If $I_R$ is the empty set then there are only $O(\eps^{-2})$ non-zero such values, and they can be computed in $\O(\log n\log \eps)$.
    Furthermore, if an interval $[a,b]\subset[i_0,i_1]$ is either added or removed to $I_R$ (maintaining disjointedness), then at most $\O(\eps^{-2})$ of these values change. Which values and how they change can be determined in $\O(\eps^{-2}\log n\log\eps^{-1})$ time.
\end{corollary}
\begin{proof}
    This is a immediate consequence of \Cref{lem:primitive2iterativesum} and  fact that $I_R$ consists of disjoint intervals and hence so does $[i_0,i_1]\setminus I_R$. Thus setting $m_{(s_k,t_k),I,I_R}=\sum_{J\in [i_0,i_1]\setminus I_R}m_{(s_k,t_k),I,J}$ and $b_{(s_k,t_k),I,I_R}=\sum_{J\in [i_0,i_1]\setminus I_R}b_{(s_k,t_k),I,J}$ implies the claim.
\end{proof}

\begin{lemma}\label{cor:finito}
    Let $I_R$ be a set of disjoint intervals each contained on a single edge of $P$. Let $e$ be a given edge and let $\sweepseq=\{(s_1,t_1),\ldots,(s_m,t_m)\}\in\edgeseqs{e}$. be a sweep-sequence. There are values $m_{-,(s,t),I_R}$,$m_{+,(s,t),I_R},b_{-,(s,t),I_R}$ and $b_{+,(s,t),I_R}$ for every $(s,t)\in \sweepseq$ such that for consecutive $(s_r,t_r),(s_{r+1},t_{r+1})$ in $\sweepseq$ it holds that
    \begin{align*}
        \sum_{(i,j)\in\nonprob(e[s_r,t_r])}&\left(\hat{r}_j(t_{r},I_R)-\hat{l}_i(s_{r},I_R)\right) = \sum_{(i,j)\in\nonprob(e[s_{r+1},t_{r+1}])}\left(\hat{r}_j(t_{r+1},I_R)-\hat{l}_i(s_{r+1},I_R)\right)\\
        &+ \left(\sum_{k\leq r}m_{+,(s_k,t_k),I_R}\right)(t_{r+1}-t_{r})+\sum_{k\leq s}b_{+,(s_k,t_k),I_R}\\
        &- \left(\sum_{k\leq r}m_{-,(s_k,t_k),I_R}\right)(s_{r+1}-s_{r})-\sum_{k\leq s}b_{-,(s_k,t_k),I_R}.
    \end{align*}
    \[\]
        If $I_R$ is the empty set then there are only $O(n\eps^{-2})$ non-zero such values, and they can be computed in $\O(n\eps^{-2} \log n\log\eps)$.
    Furthermore, if an interval $[a,b]\subset[i_0,i_1]$ is either added or removed to $I_R$ (maintaining disjointedness), then at most $\O(\eps^{-2})$ of these values change. Which values and how they change can be determined in $\O(\eps^{-2}\log n\log\eps^{-1})$ time.
\end{lemma}
\begin{proof}
By \Cref{cor:intervalset} it holds that there are $\O(n)$ intervals in $\sweepseq$ partitioned into sets $I_{1,-},I_{1,+},I_{2,-},\ldots,I_{n,+}$ that can be determined beforehand in $\O(n\log n)$ time such that
\begin{align*}
    &\sum_{(i,j)\in\nonprob(e[s_r,t_r])}\left(\hat{r}_j(t_{r+1},I_R)-\hat{l}_i(s_{r+1},I_R)\right) \\
    & = \sum_{i=1}^n\left(\sum_{I\in I_{i,+}}\mathbbm{1}_{(s_r,t_r)\in I}r_i(t_r,I_R)-\sum_{I\in I_{i,-}}\mathbbm{1}_{(s_r,t_r)\in I}l_i(s_r,I_R)\right).
\end{align*}
This together with \Cref{lem:iterativesum} imply the claim.

\end{proof}

\begin{lemma}\label{lem:finito}
    One can evaluate $\|\proxyCov_{A_S}(e[s,t])\|$ for every Type $(II)$ and $(III)$ subcurve $e[s,t]$ of $S$ in total time $\O(n^2\eps^{-2}\log^2 n\log\eps^{-1})$. Further throughout the algorithm one can maintain a data structure in total time $\O(|R|n^2\eps^{-2}\log^2 n\log\eps^{-1})$ (for all $|R|$ rounds together) that at any point correctly computes
    \[\left\|\proxyCov_{A_S}(e[s,t])\setminus\bigcup_{r\in R}\proxyCov_{A_S}(r)\right\|,\]
    for the partial solutions in time
    $\O(n^2\log^2 n)$ for every $e[s,t]$ at once.
\end{lemma}
\begin{proof}
    This is an immediate consequence of \Cref{cor:finito} together with the fact that there are a total of $\O(n\log n)$ many different sweep-sequences.
\end{proof}

\subsection{Putting everything together}

\thmMainCoverage*

\begin{proof}
    This is an immediate consequence of \Cref{lem:evalType1} and \Cref{lem:finito} together with \Cref{thm:coverageReduction} and \Cref{thm:approxFreeSpace}.
\end{proof}

\section{Conclusions}

In this paper we extend the current research on  clustering of trajectories under the Fréchet distance by developing new algorithmic techniques for the two problems Subtrajectory Clustering and Subtrajectory Coverage Maximization. Notably, we did not further improve the size of the candidate size---which may  not be possible---but rather analysed other aspects of the underlying greedy set cover algorithm. We identified two new structures, namely the sweep-sequence and the proxy coverage which together enable efficient sweep algorithms that allow internal maintenance of the proxy coverage. This leads to the first deterministic cubic $(\O(\log n),\O(1))$-approximation algorithm for the SC problem. Depending on the size of the optimal solution, the running time may even be subcubic.

Conradi and Driemel \cite{conradi2023findingcomplexpatternstrajectory} observed that previous algorithms appear practical with real-world data. Compared to their work, ours incurs an additional factor of $4$ in the approximation guarantee for the solution size. It would be interesting to see whether our algorithm behaves similarly on real-world data and whether our techniques yield results of similar quality.

It is an open problem if any algorithm for the SC problem requires at least quadratic dependence in $n$. It is tempting to assume so due to the seemingly necessary computation of the $\Delta$-free space---a central tool for almost all Fréchet distance related problems. Recall that for the clustering formulation introduced by Buchin et al. in \cite{BuchinBGLL11} Gudmundsson et al. presented a tight conditional lower bound stating that determining the `largest' cluster takes at least \textit{cubic} (in $n$) time \cite{GudmundssonW22}. This cubic lower bound suggests that any greedy set cover algorithm also requires cubic time. Interestingly, our subcubic dependency in $n$ does not contradict a lower bound of the form `finding the largest cluster requires cubic time'. This is due to the fact that at no point we compute the `largest' cluster with respect to $\Cov(\cdot)$, but rather an `approximately largest' cluster with respect to $\proxyCov(\cdot)$. It would further be interesting to see whether non-trivial lower bounds for the SC and SCM problems exist or whether an approximation algorithm for SC with quadratic dependence in $n$ exists.

With regards to the SCM problem we observe that our algorithm is as fast as computing the free space in every round of the greedy algorithm. While it would be interesting to improve the dependence in $k$ and $l$, it is likely that the quadratic dependency in $n$ is tight in the realm of current techniques that rely on the explicit computation of the free space.

\bibliography{bibliography}{}

\begin{thebibliography}{10}

\bibitem{agarwal2018}
Pankaj~K. Agarwal, Kyle Fox, Kamesh Munagala, Abhinandan Nath, Jiangwei Pan, and Erin Taylor.
\newblock Subtrajectory clustering: Models and algorithms.
\newblock In {\em Proceedings of the 37th ACM SIGMOD-SIGACT-SIGAI Symposium on Principles of Database Systems}, PODS '18, page 75–87, 2018.
\newblock \href {https://doi.org/10.1145/3196959.3196972} {\path{doi:10.1145/3196959.3196972}}.

\bibitem{AltG95}
Helmut Alt and Michael Godau.
\newblock Computing the {F}r{\'{e}}chet distance between two polygonal curves.
\newblock {\em Int. J. Comput. Geom. Appl.}, 5:75--91, 1995.
\newblock \href {https://doi.org/10.1142/S0218195995000064} {\path{doi:10.1142/S0218195995000064}}.

\bibitem{Brüning2022Faster}
Frederik Br\"{u}ning, Jacobus Conradi, and Anne Driemel.
\newblock {Faster Approximate Covering of Subcurves Under the Fr\'{e}chet Distance}.
\newblock In {\em 30th Annual European Symposium on Algorithms (ESA 2022)}, volume 244 of {\em Leibniz International Proceedings in Informatics (LIPIcs)}, pages 28:1--28:16, Dagstuhl, Germany, 2022. Schloss Dagstuhl -- Leibniz-Zentrum f{\"u}r Informatik.
\newblock \href {https://doi.org/10.4230/LIPIcs.ESA.2022.28} {\path{doi:10.4230/LIPIcs.ESA.2022.28}}.

\bibitem{Akitaya2021Covering}
Frederik Brüning, Hugo Akitaya, Erin Chambers, and Anne Driemel.
\newblock Subtrajectory clustering: Finding set covers for set systems of subcurves.
\newblock {\em Computing in Geometry and Topology}, 2(1):1:1–1:48, Feb. 2023.
\newblock URL: \url{https://www.cgt-journal.org/index.php/cgt/article/view/7}, \href {https://doi.org/10.57717/cgt.v2i1.7} {\path{doi:10.57717/cgt.v2i1.7}}.

\bibitem{buchinGroup2017}
Kevin Buchin, Maike Buchin, David Duran, Brittany~Terese Fasy, Roel Jacobs, Vera Sacristan, Rodrigo~I. Silveira, Frank Staals, and Carola Wenk.
\newblock Clustering trajectories for map construction.
\newblock In {\em Proceedings of the 25th ACM SIGSPATIAL International Conference on Advances in Geographic Information Systems}, SIGSPATIAL '17, 2017.
\newblock \href {https://doi.org/10.1145/3139958.3139964} {\path{doi:10.1145/3139958.3139964}}.

\bibitem{BuchinBGHSSSSSW20}
Kevin Buchin, Maike Buchin, Joachim Gudmundsson, Jorren Hendriks, Erfan~Hosseini Sereshgi, Vera Sacrist{\'{a}}n, Rodrigo~I. Silveira, Jorrick Sleijster, Frank Staals, and Carola Wenk.
\newblock Improved map construction using subtrajectory clustering.
\newblock In {\em LocalRec'20: Proceedings of the 4th {ACM} {SIGSPATIAL} Workshop on Location-Based Recommendations, Geosocial Networks, and Geoadvertising, LocalRec@SIGSPATIAL 2020, November 3, 2020, Seattle, WA, {USA}}, pages 5:1--5:4, 2020.
\newblock \href {https://doi.org/10.1145/3423334.3431451} {\path{doi:10.1145/3423334.3431451}}.

\bibitem{BuchinBGLL11}
Kevin Buchin, Maike Buchin, Joachim Gudmundsson, Maarten L{\"{o}}ffler, and Jun Luo.
\newblock Detecting commuting patterns by clustering subtrajectories.
\newblock {\em International Journal of Computational Geometry and Applications}, 21(3):253--282, 2011.
\newblock \href {https://doi.org/10.1142/S0218195911003652} {\path{doi:10.1142/S0218195911003652}}.

\bibitem{buchinGroup20}
Maike Buchin, Bernhard Kilgus, and Andrea Kölzsch.
\newblock Group diagrams for representing trajectories.
\newblock {\em International Journal of Geographical Information Science}, 34(12):2401--2433, 2020.
\newblock \href {https://doi.org/10.1080/13658816.2019.1684498} {\path{doi:10.1080/13658816.2019.1684498}}.

\bibitem{conradi2023findingcomplexpatternstrajectory}
Jacobus Conradi and Anne Driemel.
\newblock Finding complex patterns in trajectory data via geometric set cover, 2023.
\newblock URL: \url{https://arxiv.org/abs/2308.14865}, \href {https://arxiv.org/abs/2308.14865} {\path{arXiv:2308.14865}}.

\bibitem{de2013fast}
Mark {de Berg}, Atlas~F. Cook, and Joachim Gudmundsson.
\newblock {F}ast {F}r{é}chet queries.
\newblock {\em Computational Geometry}, 46(6):747--755, 2013.
\newblock \href {https://doi.org/10.1016/j.comgeo.2012.11.006} {\path{doi:10.1016/j.comgeo.2012.11.006}}.

\bibitem{frederickson1984matrix}
Greg~N. Frederickson and Donald~B. Johnson.
\newblock Generalized selection and ranking: Sorted matrices.
\newblock {\em SIAM Journal on Computing}, 13(1):14--30, 1984.
\newblock \href {https://arxiv.org/abs/https://doi.org/10.1137/0213002} {\path{arXiv:https://doi.org/10.1137/0213002}}, \href {https://doi.org/10.1137/0213002} {\path{doi:10.1137/0213002}}.

\bibitem{GudmundssonV15}
Joachim Gudmundsson and Nacho Valladares.
\newblock A {GPU} approach to subtrajectory clustering using the {F}r{\'{e}}chet distance.
\newblock {\em {IEEE} Trans. Parallel Distributed Syst.}, 26(4):924--937, 2015.
\newblock \href {https://doi.org/10.1109/TPDS.2014.2317713} {\path{doi:10.1109/TPDS.2014.2317713}}.

\bibitem{abs-2110-15554}
Joachim Gudmundsson and Sampson Wong.
\newblock Cubic upper and lower bounds for subtrajectory clustering under the continuous {F}réchet distance, 2021.
\newblock \href {https://doi.org/10.48550/ARXIV.2110.15554} {\path{doi:10.48550/ARXIV.2110.15554}}.

\bibitem{GudmundssonW22}
Joachim Gudmundsson and Sampson Wong.
\newblock Cubic upper and lower bounds for subtrajectory clustering under the continuous fr{\'{e}}chet distance.
\newblock In Joseph~(Seffi) Naor and Niv Buchbinder, editors, {\em Proceedings of the 2022 {ACM-SIAM} Symposium on Discrete Algorithms, {SODA} 2022, Virtual Conference / Alexandria, VA, USA, January 9 - 12, 2022}, pages 173--189. {SIAM}, 2022.
\newblock \href {https://doi.org/10.1137/1.9781611977073.9} {\path{doi:10.1137/1.9781611977073.9}}.

\bibitem{ionescu2013human3}
Catalin Ionescu, Dragos Papava, Vlad Olaru, and Cristian Sminchisescu.
\newblock Human3.6m: Large scale datasets and predictive methods for 3d human sensing in natural environments.
\newblock {\em IEEE transactions on pattern analysis and machine intelligence}, 36(7):1325--1339, 2013.
\newblock \href {https://doi.org/10.1109/TPAMI.2013.248} {\path{doi:10.1109/TPAMI.2013.248}}.

\bibitem{krause11}
Andreas Krause and Daniel Golovin.
\newblock Submodular function maximization.
\newblock In Lucas Bordeaux, Youssef Hamadi, and Pushmeet Kohli, editors, {\em Tractability: Practical Approaches to Hard Problems}, pages 71--104. Cambridge University Press, 2014.
\newblock \href {https://doi.org/10.1017/CBO9781139177801.004} {\path{doi:10.1017/CBO9781139177801.004}}.

\bibitem{lee2002gait}
Lily Lee and W~Eric~L Grimson.
\newblock Gait analysis for recognition and classification.
\newblock In {\em Proceedings of Fifth IEEE International Conference on Automatic Face Gesture Recognition}, pages 155--162. IEEE, 2002.

\bibitem{LiangYWLCXL24}
Anqi Liang, Bin Yao, Bo~Wang, Yinpei Liu, Zhida Chen, Jiong Xie, and Feifei Li.
\newblock Sub-trajectory clustering with deep reinforcement learning.
\newblock {\em {VLDB} J.}, 33(3):685--702, 2024.
\newblock URL: \url{https://doi.org/10.1007/s00778-023-00833-w}, \href {https://doi.org/10.1007/S00778-023-00833-W} {\path{doi:10.1007/S00778-023-00833-W}}.

\bibitem{Nemhauser1978}
G.~L. Nemhauser, L.~A. Wolsey, and M.~L. Fisher.
\newblock An analysis of approximations for maximizing submodular set functions—{I}.
\newblock {\em Mathematical Programming}, 14(1):265--294, 1978.
\newblock \href {https://doi.org/10.1007/BF01588971} {\path{doi:10.1007/BF01588971}}.

\bibitem{Qiao2017RealtimeHG}
Sen Qiao, Y.~Wang, and J.~Li.
\newblock Real-time human gesture grading based on {O}pen{P}ose.
\newblock {\em 2017 10th International Congress on Image and Signal Processing, BioMedical Engineering and Informatics (CISP-BMEI)}, pages 1--6, 2017.
\newblock \href {https://doi.org/10.1109/CISP-BMEI.2017.8301910} {\path{doi:10.1109/CISP-BMEI.2017.8301910}}.

\bibitem{acmsurvey20}
Roniel S.~De Sousa, Azzedine Boukerche, and Antonio A.~F. Loureiro.
\newblock Vehicle trajectory similarity: Models, methods, and applications.
\newblock {\em ACM Comput. Surv.}, 53(5), September 2020.
\newblock \href {https://doi.org/10.1145/3406096} {\path{doi:10.1145/3406096}}.

\bibitem{vanderhoog2024fasterdeterministicsubtrajectoryclustering}
Ivor van~der Hoog, Thijs van~der Horst, and Tim Ophelders.
\newblock Faster and deterministic subtrajectory clustering, 2024.
\newblock URL: \url{https://arxiv.org/abs/2402.13117}, \href {https://arxiv.org/abs/2402.13117} {\path{arXiv:2402.13117}}.

\bibitem{VANSEBILLE201849}
Erik {van Sebille}, Stephen~M. Griffies, Ryan Abernathey, Thomas~P. Adams, Pavel Berloff, Arne Biastoch, Bruno Blanke, Eric~P. Chassignet, Yu~Cheng, Colin~J. Cotter, Eric Deleersnijder, Kristofer Döös, Henri~F. Drake, Sybren Drijfhout, Stefan~F. Gary, Arnold~W. Heemink, Joakim Kjellsson, Inga~Monika Koszalka, Michael Lange, Camille Lique, Graeme~A. MacGilchrist, Robert Marsh, C.~Gabriela {Mayorga Adame}, Ronan McAdam, Francesco Nencioli, Claire~B. Paris, Matthew~D. Piggott, Jeff~A. Polton, Siren Rühs, Syed~H.A.M. Shah, Matthew~D. Thomas, Jinbo Wang, Phillip~J. Wolfram, Laure Zanna, and Jan~D. Zika.
\newblock Lagrangian ocean analysis: Fundamentals and practices.
\newblock {\em Ocean Modelling}, 121:49--75, 2018.
\newblock URL: \url{https://www.sciencedirect.com/science/article/pii/S1463500317301853}, \href {https://doi.org/10.1016/j.ocemod.2017.11.008} {\path{doi:10.1016/j.ocemod.2017.11.008}}.

\bibitem{HotspotSurvey}
Yiqun Xie, Shashi Shekhar, and Yan Li.
\newblock Statistically-robust clustering techniques for mapping spatial hotspots: A survey.
\newblock {\em ACM Comput. Surv.}, 55(2), January 2022.
\newblock \href {https://doi.org/10.1145/3487893} {\path{doi:10.1145/3487893}}.

\end{thebibliography}
\bibliographystyle{plainurl}

\addtocounter{linenumber}{-1000}
\appendix

\section{Proofs of \Cref{sec:setsystem}}\label{apx:setsystemproofs}

\lemmaapxOne*
\begin{proof}
    Let $0\leq s\leq t\leq 1$ such that $\hat{\pi}=e[s,t]$. Let $\eps_1$ be the largest value among $\extremal{}(A_e)$ such that $\eps_1\leq s$. Let $\eps_2$ be the smallest value among $\extremal{}(A_e)$ such that $\eps_2\geq s$. Let similarly $\eps_3$ be the largest value among $\extremal{}(A_e)$ such that $\eps_3\leq t$ and $\eps_4$ be the smallest value among $\extremal{}(A_e)$ such that $\eps_4\geq t$. 
    
    Assume first that $\eps_2\leq \eps_3$ Observe that since there is no extremal point of $A_e$ between $s$ and $\eps_1$, and between $s$ and $\eps_2$, and there is no extremal point of $A_e$ between $t$ and $\eps_3$, and between $t$ and $\eps_4$, If there is a monotone path from some cell $i$ at height $s$ to cell $j$ at height $t$, then there are also monotone paths from cell $i$ at height $\eps_1$ (resp. $\eps_2$) to cell $j$ at height $\eps_3$ (resp. $\eps_4$). By the convexity of the free-space in every cell and the fact that the $y$-coordinates of the left-most points are also in $\extremal{}(A_e)$ it holds further, that $\min(l_i(\eps_1),l_i(\eps_2))\leq l_i(s)$ and $r_i(t)\leq \max(r_j(\eps_3),r_j(\eps_4))$, where $l_i$ and $r_j$ are defined as in \Cref{def:cellfunc}. And thus 
    \[\Cov_P(\hat{\pi},4\Delta)\subset \Cov_{A_S}(\{e[\eps_1,\eps_3],e[\eps_1,\eps_4],e[\eps_2,\eps_3],e[\eps_2,\eps_4]\}).\]
    If $s$ itself is an extremal coordinate, then $\eps_1=\eps_2$. Similarly if $t$ itself is an extremal coordinate, then $\eps_3=\eps_4$. This is the case if $\hat{\pi}$ is an edge-affix. In this case the coverage $\Cov_P(\hat{\pi},4\Delta)$ is a subset of either $\Cov_{A_S}(\{e[\eps_1,\eps_3],e[\eps_1,\eps_4]\})$ or $\Cov_{A_S}(\{e[\eps_1,\eps_3],e[\eps_2,\eps_3]\})$.
    
    If instead $\eps_2 > \eps_3$ then in particular $\eps_1 = \eps_3 < \eps_2=\eps_4$. But then since there is no extremal point of $A_e$ between $s$ and $\eps_1$, and between $s$ and $\eps_2$, and there is no extremal point of $A_e$ between $t$ and $\eps_3$, and between $t$ and $\eps_4$, If there is a monotone path from some cell $i$ at height $s$ to cell $j$ at height $t$, then there are also monotone paths from cell $i$ at height $\eps_1$ to cell $j$ at height $\eps_3$. There are further monotone paths from cell $i$ at height $\eps_1$ to cell $j$ at height $\eps_4$ and from cell $i$ at height $\eps_2$ to cell $j$ at height $\eps_4$. Lastly there is also a path from cell $i$ at height $\eps_3$ to cell $j$ at height $\eps_1$ that is monotonously increasing in the $x$-coordinate and monotonously \textit{decreasing} in the $y$-coordinate. But this implies that 
    \[\Cov_P(\hat{\pi},4\Delta)\subset \Cov_{A_S}(\{e[\eps_1,\eps_3],e[\eps_1,\eps_4],\rev{e[\eps_3,\eps_2]},e[\eps_2,\eps_4]\}).\]
    If $s$ itself is an extremal coordinate, then $\eps_1=\eps_2$. Similarly if $t$ itself is an extremal coordinate, then $\eps_3=\eps_4$. This is the case if $\hat{\pi}$ is an edge-affix. In this case the coverage $\Cov_P(\hat{\pi},4\Delta)$ is a subset of either $\Cov_{A_S}(\{e[\eps_1,\eps_3],e[\eps_1,\eps_4]\})$ or $\Cov_{A_S}(\{e[\eps_1,\eps_3],\rev{e[\eps_3,\eps_2]}\})$, which concludes the proof.
    
\end{proof}

\lemmaapxTwo*
\begin{proof}
    Let $\pi$ be given. By \Cref{lem:vanderhoogsplit} and \Cref{lem:foursplit} there is a set $S_\pi'$ consisting of either four edge-affixes starting and ending at values in $\extremal{}(e,P)$ of some edge $e$---that is to say, they are Type $(II)$-subcurves of $S$---and one vertex-vertex subcurve of $S$ of complexity $\ell$, or four subedges of $S$ starting and ending at values in $\extremal{}(e,P)$. It suffices to show that any vertex-vertex subcurve of $S$ of complexity at most $\ell$ is dominated by two $(I)$-subcurves, and any subedge of $S$ starting and ending at values in $\extremal{}(e,P)$ is dominated by two $(III)$-subcurves.

    A vertex-vertex subcurve $S[s,t]$ starting at vertex $i$ and ending at vertex $i+j$ of $S$ splits into two `overlapping' $(I)$-subcurves, namely $S[s,r]$ from vertex $i$ to vertex $r=i+2^{\lfloor\log_2(j)\rfloor}$ and $S[l,t]$ from vertex $l=j-2^{\lfloor\log_2(j)\rfloor}$ to vertex $j$. As $s\leq l\leq r\leq t$, any interval $[a,b]$ such that there is a monotone path from $(a,s)$ to $(b,t)$ similarly splits into $[a,r']$ and $[l',b]$, such that $a\leq l'\leq r'\leq b$ and there are monotone paths from $(a,s)$ to $(r',r)$ and from $(l',l)$ to $(b,t)$, and thus 
    \[[a,b]=[a,r']\cup[l',r]\subset\Cov_{A_S}(S[s,r])\cup\Cov_{A_S}(S[l,t]).\]
    A similar argument shows that any subedge splits into two `overlapping' $(III)$-subcurves, which concludes the proof.
\end{proof}
\end{document}